\documentclass{amsart}

\usepackage[T1]{fontenc}
\usepackage[latin1]{inputenc}
\usepackage{url}
\usepackage[colorlinks=true]{hyperref}
\usepackage{dsfont}
\usepackage{graphicx}
\usepackage{enumitem}

\renewcommand{\hat}[1]{\ensuremath{\widehat{#1}}}
\newcommand{\Sim}{\ensuremath{\mathop{\sim}\limits}}
\newcommand{\grandO}{\ensuremath{\mathop{\mathrm{O}}\limits}}
\newcommand{\crochet}[2]{\left\langle #1, #2\right\rangle}
\newcommand{\scrochet}[2]{\langle #1, #2\rangle}

\newcommand{\ie}{{\em i.e.}~}
\newcommand{\dd}{\mathrm{d}}
\newcommand{\R}{\mathbb{R}}
\renewcommand{\P}{\mathcal{P}}
\newcommand{\Exp}{\mathbb{E}}
\renewcommand{\Pr}{\mathbb{P}}
\newcommand{\ind}[1]{\mathds{1}_{\{#1\}}}
\newcommand{\ua}{\underline{a}}
\newcommand{\pc}{p_{\mathrm{c}}}
\newcommand{\qc}{q_{\mathrm{c}}}
\newcommand{\Beta}{\mathrm{Beta}}
\newcommand{\Cov}{\mathrm{Cov}}
\newcommand{\ex}[1]{\mathrm{e}^{#1}}
\newcommand{\uto}{\uparrow}
\newcommand{\dto}{\downarrow}
\newcommand{\tY}{\tilde{Y}}
\newcommand{\tF}{\tilde{F}}
\newcommand{\Z}{\mathcal{Z}}
\newcommand{\bZ}{\bar{\mathcal{Z}}}
\newcommand{\bPi}{\bar{\Pi}}
\newcommand{\bmu}{\bar{\mu}}

\newtheorem{defi}{Definition}[section]
\newtheorem{lem}[defi]{Lemma}
\newtheorem{prop}[defi]{Proposition}
\newtheorem{theo}[defi]{Theorem}
\newtheorem{cor}[defi]{Corollary}

\newtheoremstyle{myremark}{}{}{}{0pt}{\bfseries}{.}{ }{}
\theoremstyle{myremark}
\newtheorem{rk}[defi]{Remark}
\newtheorem{expl}[defi]{Example}

\title{Capital distribution and portfolio performance in the mean-field Atlas model}

\author{Benjamin Jourdain}
\address{{\bf Benjamin Jourdain}\newline
{\rm Université Paris-Est, CERMICS (ENPC), INRIA, F-77455 Marne-la-Vallée}}
\email{\href{mailto:jourdain@cermics.enpc.fr}{jourdain@cermics.enpc.fr}}

\author{Julien Reygner}
\address{{\bf Julien Reygner}\newline
{\rm Sorbonne Universités, UPMC Univ Paris 06, UMR 7599, LPMA, F-75005 Paris\newline
Université Paris-Est, CERMICS (ENPC), F-77455 Marne-la-Vallée}}
\email{\href{mailto:julien.reygner@polytechnique.org}{julien.reygner@polytechnique.org}}

\thanks{This research benefited from the support of the Chaire Risques Financiers, Fondation du Risque, and of the French National Research Agency (ANR) under the program ANR-12-BLAN Stab.}


\begin{document}

\begin{abstract}
  We study a mean-field version of rank-based models of equity markets such as the Atlas model introduced by Fernholz in the framework of Stochastic Portfolio Theory. We obtain an asymptotic description of the market when the number of companies grows to infinity. Then, we discuss the long-term capital distribution. We recover the Pareto-like shape of capital distribution curves usually derived from empirical studies, and provide a new description of the phase transition phenomenon observed by Chatterjee and Pal. Finally, we address the performance of simple portfolio rules and highlight the influence of the volatility structure on the growth of portfolios.
\end{abstract}

\maketitle

\section{Introduction}\label{s:intro}

\subsection{Rank-based models}\label{ss:intro:rbm} Rank-based models of equity markets were introduced by Fernholz within the framework of {\em Stochastic Portfolio Theory}~\cite{fernholz, ferkar} as first-order approximations of asymptotically stable markets. In such models, the capitalization of a stock is described by the exponential of a diffusion process, the drift and variance of which depend only on the rank of the stock among the whole market. A simple but celebrated instance of such a model is the {\em Atlas model}~\cite{fernholz, banner, ferkar, ipbkf}, where all the stocks have the same variance and the smallest stock is responsible for the growth of the whole market.

In the long-term, the Atlas model was proven to capture the actual distribution of the total capital~\cite{fernholz}. This gave rise to a large amount of mathematical studies on rank-based models~\cite{ik, iks, fikp, fik:skew}; in particular, concerning the shape of capital distribution curves~\cite{banner, chapalPTRF, ferkar, sarantsev} as well as the selection of optimal investment strategies (portfolios) on the market~\cite{banner, ferkar}. Both the capital distribution and the performance of portfolios depend on the long time behaviour of the market, which was described in~\cite{banner, pp, jm, ips, ipbkf}. In order to study large markets, asymptotic properties, when the number of stocks grows to infinity, of long-term rank-based models were derived in~\cite{banner, chapalPTRF, sarantsev}.

In this article, we introduce a rank-based model that we call the {\em mean-field Atlas model}, where the drift and variance of the capitalization processes depend on empirical quantiles. This particular shape for the characteristics of the market, that we shall discuss below, allows us to:
\begin{enumerate}
  \item derive an asymptotic description of the evolution of the market when its size grows to infinity, through a functional law of large numbers;
  \item obtain closed form expressions for the long time behaviour of this asymptotic market;
  \item recover capital distribution curves similar to those empirically observed;
  \item carry out a detailed analysis of the performance of a portfolio rule.
\end{enumerate}

Before providing more insight into these issues in Subsection~\ref{ss:intro:context} and giving a proper definition of our model in Subsection~\ref{ss:intro:model}, let us insist on the following particularity of our approach. In all the works cited above, the authors first address the long time behaviour of market models with a fixed number of stocks, then possibly study the large size limit of the market under its steady state. The latter is not so easy to handle as the underlying stationary distribution is generically not known, see~\S\ref{sss:intro:lt} below for a more detailed review. As a consequence, the asymptotic behaviour of these steady states for large markets is all the more difficult to understand, although there have been remarkable results in this direction~\cite{banner, chapalPTRF}.

In the present paper, we somehow take the opposite path and first obtain an asymptotic description of the evolution of the whole market when the number of stocks grows to infinity. This limit shall be referred to as the {\em asymptotic market}. Then, we address the long time behaviour of this asymptotic market and get an explicit description of the steady states of large markets, which is widely based on the theoretical results of~\cite{jourey}. To our knowledge, this is the first study proceeding in this way. For models assigning the same volatility to each stock, we provide a mathematical justification of the equivalence of both approaches, thanks to the technical results of~\cite{chaoticity}.

\subsection{Context and motivations}\label{ss:intro:context} We now provide a general introduction to the issues we shall address in the context of the mean-field Atlas model; namely, the long-term stability of rank-based models, the description of capital distribution curves and the analysis of portfolio performance.

\subsubsection{Long-term stability of rank-based models}\label{sss:intro:lt} The framework of Stochastic Portfolio Theory~\cite{fernholz, ferkar} is described as follows. For a market containing a fixed number $n \geq 1$ of stocks, with respective capitalizations $X^1_n(t), \ldots, X_n^n(t) > 0$ at time $t$, the log-capitalizations $Y^i_n(t) := \log X^i_n(t)$ are assumed to satisfy the relation
\begin{equation}\label{eq:model}
  \forall i \in \{1, \ldots, n\}, \qquad \dd Y^i_n(t) = \gamma^i_n(t) \dd t + \sigma^i_n(t) \dd B^i(t),
\end{equation}
where the {\em growth rate} process $(\gamma^1_n(t), \ldots, \gamma^n_n(t))_{t \geq 0}$ and the {\em volatility} process $(\sigma^1_n(t), \ldots, \sigma^n_n(t))_{t \geq 0}$ in $\R^n$ are adapted to a given filtration $(\mathcal{F}(t))_{t \geq 0}$ on some probability space $(\Omega,\mathcal{F},\Pr)$, and the processes $(B^i(t))_{t \geq 0}$, $i \geq 1$ are independent $(\mathcal{F}(t))_{t \geq 0}$-Brownian motions.

The model is said to be {\em rank-based} whenever the growth rate process and volatility process write 
\begin{equation}\label{eq:gammasigma}
  \gamma^i_n(t) = \sum_{j=1}^n \ind{Y^i_n(t) = Y^{(j)}_n(t)} \gamma^j_n, \qquad \sigma^i_n(t) = \sum_{j=1}^n \ind{Y^i_n(t) = Y^{(j)}_n(t)}\sigma^j_n,
\end{equation}
for given growth rate coefficients $\gamma^1_n, \ldots, \gamma^n_n \in \R$ and volatility coefficients $\sigma^1_n, \ldots, \sigma^n_n \in \R$, where $Y^{(1)}_n(t) \leq \cdots \leq Y^{(n)}_n(t)$ refer to the increasing reordering of $Y^1_n(t), \ldots, Y^n_n(t)$. In other words, the dynamics of each stock is determined by its rank among the whole market. As soon as, for all $j \in \{1, \ldots, n\}$, $(\sigma^j_n)^2 > 0$, then the stochastic differential equation~(\ref{eq:model},~\ref{eq:gammasigma}) admits a unique weak solution~\cite{bass}, and almost surely, $\dd t$-almost everywhere, the random variables $Y^1_n(t), \ldots, Y^n_n(t)$ are pairwise distinct, therefore there is no need to take a specific convention to resolve ties. Then, we define the capitalization of the $i$-th stock by $X^i_n(t) := \exp Y^i_n(t)$.

Let us emphasize the fact that we use the notation $(j)$ to refer to the {\em increasing reordering}, following the usual convention for order statistics. However it is sometimes convenient to use the decreasing reordering~\cite{fernholz, ferkar, pp, chapalPTRF}. In such situations, we shall use the notation $[k]$. In other words, if $(y^1, \ldots, y^n) \in \R^n$, then $((1), \ldots, (n))$ is a permutation of $(1, \ldots, n)$ such that $y^{(1)} \leq \cdots \leq y^{(n)}$, while $([1], \ldots, [n])$ is a permutation of $(1, \ldots, n)$ such that $y^{[1]} \geq \cdots \geq y^{[n]}$. Note that one may always choose $[k] = (n-k+1)$.

A first mathematical study of rank-based models was carried out by Banner, Fernholz and Karatzas~\cite{banner}. There, the emphasis was laid on the particular choice for the growth rate coefficients
\begin{equation}\label{eq:atlas}
  \gamma^1_n = ng, \qquad \gamma^2_n = \cdots = \gamma^n_n = 0,
\end{equation}
where $g > 0$. With this choice of coefficients, the smallest stock is responsible for the growth of the whole market, therefore, analogically to the ancient Greek myth, this model is called the {\em Atlas model}. Various generalizations of this model were introduced later, such as {\em hybrid Atlas models} by Ichiba, Papathanakos, Banner, Karatzas and Fernholz~\cite{ipbkf} (see also Fernholz, Ichiba and Karatzas~\cite{fik}), in which the growth rate and volatility processes depend both on the rank and on the index $i$ of a stock.

As far as the long time behaviour of the solution to~(\ref{eq:model},~\ref{eq:gammasigma}) is concerned, Banner, Fernholz and Karatzas~\cite{banner} described the marginal distribution of each stock in the long-term. Pal and Pitman~\cite{pp} and Jourdain and Malrieu~\cite{jm} described their joint distribution in the long-term for models in which all the stocks are assigned the same variance, and Ichiba, Papathanakos, Banner, Karatzas and Fernholz~\cite{ipbkf} extended these results to the case of a linearly decreasing variance coefficient with respect to the rank. Rates of convergence were provided by Ichiba, Pal and Shkolnikov~\cite{ips}. 

Generically, a necessary and sufficient condition ensuring the long-term stability of rank-based models~(\ref{eq:model},~\ref{eq:gammasigma}) is that
\begin{equation}\label{eq:stabAtlas}
  \forall k \in \{1, \ldots, n-1\}, \qquad \sum_{j=1}^k (\gamma^j_n - g_n) > 0,
\end{equation}
where $g_n := (1/n)\sum_{j=1}^n \gamma^j_n$ is the mean growth rate of the processes $(Y^1_n(t))_{t \geq 0}, \ldots, (Y^n_n(t))_{t \geq 0}$. This condition somehow expresses the fact that the growth rate of small stocks is larger than the mean growth rate of the market, while the growth rate of large stocks is smaller than the mean growth rate of the market. This is known as the {\em size effect}, see~\cite[p.~86]{fernholz}. From the economic point of view, this effect is a natural consequence of {\em rebalancing}, that is, the fact that investors buy stocks when their prices are low and sell stocks with large prices.

Similarly to the size effect on growth rates, the variance of small stocks is also empirically observed to be larger than the variance of large stocks. Throughout this article, we shall refer to this phenomenon as the {\em volatility size effect}. As an example, in~\cite[Figure~13.6]{ferkar}, the variance is observed to be linearly decreasing with respect to the rank, which motivates the stability result of~\cite{ipbkf}. 

Several models have been introduced to capture the growth rate and volatility size effects, see for instance the so-called {\em Volatility-Stabilized Model} by Fernholz and Karatzas~\cite[Section~12]{ferkar}, which was later on discussed by Pal~\cite{palAAP11}, Shkolnikov~\cite{shkSPA13} and Sarantsev~\cite{sarAF}. As we shall see below, both rebalancing and the volatility size effect play a key role in the analysis of portfolio performance.

\subsubsection{Capital distribution curves}\label{sss:intro:capital} For all $i \in \{1, \ldots, n\}$, the market weight at time $t \geq 0$ of the $i$-th stock is defined by
\begin{equation*}
  \mu^i_n(t) := \frac{X^i_n(t)}{X^1_n(t) + \cdots + X^n_n(t)}.  
\end{equation*}
The {\em capital distribution curve} is the logarithmic representation of the market weights rearranged by decreasing order, namely the curve $\log k \mapsto \log \mu^{[k]}_n(t)$, where we recall that $[k]$ refers to the index of the stock with $k$-th largest capitalization at time $t$; that is to say, $X^{[1]}_n(t) \geq \cdots \geq X^{[n]}_n(t)$. 

The actual capital distribution curves for the stocks traded on the NYSE, the AMEX and the NASDAQ stock market between 1929 and 1999 were described by Fernholz~\cite[Figure~5.1, p.~95]{fernholz}. They exhibit a remarkable stability over time, and indicate, at least for the largest stocks, a Pareto-like distribution of the capital, which is a common observation in the economic literature~\cite[Section~5.6]{fernholz}.

This Pareto-like distribution was recovered for the Atlas model~\eqref{eq:atlas} by Fernholz~\cite[Example~5.3.3, pp.~103-104]{fernholz}. In the case of rank-based models~(\ref{eq:model},~\ref{eq:gammasigma}) with constant variance coefficients, Chatterjee and Pal~\cite{chapalPTRF} adressed the asymptotic behaviour, when $n$ grows to infinity, of the stationary distribution of $(\mu^{[1]}_n(t), \ldots, \mu^{[n]}_n(t))_{t \geq 0}$. They observed the following {\em phase transition} phenomenon: depending on the growth rate coefficients,
\begin{itemize}
  \item either the largest stock dominates the market and monopolizes all the capital,
  \item or most of the capital is spread among a few leading stocks,
  \item or the market weight of every stock vanishes.
\end{itemize}
In the second case above, the distribution of the capital between the few leading stocks was also proven to exhibit a Pareto-like distribution.

\subsubsection{Portfolio performance}\label{sss:intro:portfolio} A portfolio rule on an equity market is a strategy prescribing the proportion of wealth to be invested in each stock. In particular, the {\em equally weighted portfolio} assigns the same weight to all stocks, while the {\em market portfolio} is given by market weights. Due to the fact that these two strategies can easily be implemented, they are of importance for practitioners and have aroused many empirical and theoretical studies. 

From the empirical point of view, it has been observed that the equally weighted portfolio generally outperforms the market portfolio (`beats the market') under various indicators; we refer to the work by Plyakha, Uppal and Vilkov~\cite{plyakha} for a study of the major U.S. equity indices over the last four decades. From the theoretical point of view, it is commonly believed that the equally weighted portfolio beating the market is due to rebalancing: indeed, the market portfolio tends to invest more capital in large stocks, while the equally weighted portfolio is insensitive to this effect; see the preprint by Pal and Wong~\cite{palwong} and the references therein. As far as the Atlas model is concerned, the performance of the equally weighted portfolio and the market portfolio was addressed by Banner, Fernholz and Karatzas~\cite{banner}, who essentially confirmed that the equally weighted portfolio beats the market.

\subsection{Model and results}\label{ss:intro:model} We complete this introduction by giving a proper definition of the mean-field Atlas model and providing an overview of our results.

\subsubsection{The mean-field Atlas model}\label{sss:intro:mfa} Let $\gamma, \sigma : [0,1] \to \R$ be continuous functions; $\gamma$ is the {\em growth rate function}, $\sigma$ is the {\em volatility function}. The function $\sigma^2$ shall be called the {\em variance function}. The mean-field Atlas model consists of the rank-based model~(\ref{eq:model},~\ref{eq:gammasigma}), with growth rate and volatility coefficients given by, for all $n \geq 1$,
\begin{equation}\label{eq:mfgammasigma}
  \forall j \in \{1, \ldots, n\}, \qquad \gamma^j_n = \gamma(j/n), \qquad \sigma^j_n = \sigma(j/n);
\end{equation}
and initial log-capitalizations $Y^1_n(0), \ldots, Y^n_n(0)$ i.i.d. according to a given probability distribution $m$ on $\R$. It is well defined as soon as $\sigma^2(u) > 0$ for all $u \in [0,1]$, which we shall refer to as the uniform ellipticity assumption~\eqref{hyp:UE} in the sequel.

For all $j \in \{1, \ldots, n\}$, for all $t \geq 0$, $Y^{(j)}_n(t)$ is the empirical quantile of order $j/n$ of the vector $(Y^1_n(t), \ldots, Y^n_n(t))$, so that the growth rate and volatility of the log-capitalization process $Y^{(j)}_n$ is a function of $j/n$. From the point of view of economical modelling, we argue that mean-field coefficients~\eqref{eq:mfgammasigma} are reasonable choices for large rank-based models as they describe {\em weak interactions} between the stocks, in the sense that the larger the market is, the smaller the individual influence of a stock on another is.

\begin{rk}\label{rk:atlas} 
  Let us emphasize that the mean-field Atlas model is not a generalization of the genuine Atlas model~\eqref{eq:atlas}: formally, to recover (in the large size limit) the growth rate coefficients~\eqref{eq:atlas} from the mean-field coefficients~\eqref{eq:mfgammasigma}, one should replace the growth rate function $\gamma$ with $g \delta_0$, where $\delta_0$ is the Dirac distribution in $0$. Of course, this is not a function and therefore the Atlas model cannot be rigorously described in terms of mean-field coefficients. However, mean-field approximations of the Atlas model can be introduced by using the growth rate function
  \begin{equation*}
    \gamma_{\alpha}(u) := g(\alpha+1)(1-u)^{\alpha}, \qquad g > 0,
  \end{equation*} 
  where $\alpha > 0$ is the {\em Atlas index}: the larger it is, the more the growth rate concentrates on small stocks. This mean-field approximation of the Atlas model is used in Example~\ref{expl:portfolio} below.
\end{rk}

\subsubsection{Results and outline of the article}\label{sss:intro:results} Section~\ref{s:meanfield} is dedicated to the description of the asymptotic behaviour of the mean-field Atlas model in the large size limit. This issue was first addressed by Shkolnikov~\cite{shkolnikov} for stationary initial distributions, and then by the authors~\cite{jourey} for generical initial distributions (see also the recent article by Dembo, Shkolnikov, Varadhan and Zeitouni~\cite{dsvz}). The following {\em propagation of chaos} phenomenon was observed: when $n$ grows to infinity, the log-capitalization processes asymptotically behave like independent copies of a stochastic process $(Y(t))_{t \geq 0}$, such that, for all $t \geq 0$, 
\begin{equation}\label{eq:EY}
  \Exp(Y(t)) = \Exp(Y(0)) + gt,
\end{equation}
where
\begin{equation}\label{def:g}
  g := \int_{u=0}^1 \gamma(u)\dd u
\end{equation}
is the {\em market mean growth rate}. In other words, the chaoticity of the i.i.d. initial conditions is asymptotically propagated to the log-capitalization processes when their number is large. We first recall this result, and then describe the long time behaviour of the fluctuation $\tY(t)$ of $Y(t)$ around $gt$. Under a size effect assumption of the same nature as~\eqref{eq:stabAtlas}, we prove that the law of $\tY(t)$ converges toward an explicit equilibrium distribution. We also discuss the shape of the tails of this equilibrium distribution in $-\infty$ and $+\infty$.

In Section~\ref{s:Pi}, we define the {\em weighted capital measure} $\Pi_n^p(t)$ by
\begin{equation}\label{eq:Pin}
  \Pi_n^p(t) := \sum_{j=1}^n \frac{(X_n^{(j)}(t))^p}{(X_n^1(t))^p + \cdots + (X_n^n(t))^p} \delta_{j/n},
\end{equation}
for all {\em diversity indices} $p \geq 0$. When $p=1$, we drop the superscript notation and only refer to $\Pi_n(t)$ as the {\em capital measure}.

The weighted capital measure is a random probability measure on $[0,1]$. Our study of capital distribution curves and portfolio performance is based on the analysis of $\Pi_n^p(t)$ when $n \to +\infty$ and $t \to +\infty$. We first use our propagation of chaos result to derive a law of large numbers for $\Pi_n^p(t)$; namely, we prove that
\begin{equation*}
  \lim_{n \to +\infty} \Pi_n^p(t) = \Pi^p(t),
\end{equation*}
where the {\em asymptotic weighted capital measure} $\Pi^p(t)$ is a deterministic probability distribution on $[0,1]$, with an explicit expression in terms of the law of $Y(t)$. Then, we address the long time behaviour of $\Pi^p(t)$, and prove that there exists a {\em critical diversity index} $\pc \geq 0$ such that:
\begin{itemize}
  \item if $p \in [0,\pc)$, then
  \begin{equation*}
    \lim_{t \to +\infty} \Pi^p(t) = \bPi^p,
  \end{equation*}
  where the {\em long-term asymptotic weighted capital measure} $\bPi^p$ is a probability distribution on $[0,1]$, with an explicit density with respect to the Lebesgue measure on $[0,1]$, given in terms of the equilibrium distribution introduced above,
  \item if $p > \pc$, then
  \begin{equation*}
    \lim_{t \to +\infty} \Pi^p(t) = \delta_1.
  \end{equation*}
\end{itemize}
We shall refer to the fact that the model behave differently whether $p \in [0,\pc)$ or $p>\pc$ as a {\em phase transition} phenomenon, and the case $p \in [0,\pc)$ (resp. $p=\pc$ and $p > \pc$) shall be called the {\em subcritical phase} (resp. {\em criticality} and the {\em supercritical phase}). 

We conclude the study of the weighted capital measure by discussing the validity of the interversion property
\begin{equation*}
  \lim_{t \to +\infty} \lim_{n \to +\infty} \Pi^p_n(t) = \lim_{n \to +\infty} \lim_{t \to +\infty} \Pi^p_n(t)
\end{equation*}
in Subsection~\ref{ss:interv}.

In Section~\ref{s:capital}, we study the distribution of the capital for the long-term asymptotic market. This relies on the analysis of the capital measure as follows. Recall from~\S\ref{sss:intro:capital} that the capital distribution curve describes the repartition of capital with respect to the rank of companies, ordered by size. For the sake of coherence with the works by Fernholz~\cite{fernholz} and Chatterjee and Pal~\cite{chapalPTRF}, the companies are ranked with respect to the decreasing order of their size: $\mu^{[1]}_n(t) \geq \cdots \geq \mu^{[n]}_n(t)$. We recall that one can choose $[k] = (n-k+1)$.

For $u,v \in [0,1]$ with $u \leq v$, the proportion of capital held by companies ranked between $nu$ and $nv$ is roughly
\begin{equation*}
  \sum_{nu \leq k \leq nv} \mu^{[k]}_n(t) = \sum_{nu \leq k \leq nv} \mu^{(n-k+1)}_n(t) \simeq \sum_{n(1-v) \leq j \leq n(1-u)} \mu^{(j)}_n(t) = \crochet{\ind{1-v \leq \cdot \leq 1-u}}{\Pi_n(t)},
\end{equation*}
which explicits the link between the capital distribution curves and the capital measure $\Pi_n(t)$. In order to describe the long-term capital distribution on large markets, we use the results of Section~\ref{s:Pi} on the long time behaviour of the asymptotic capital measure $\Pi(t)$.

Interestingly, the phase transition for the asymptotic weighted capital measure derived in Section~\ref{s:Pi} results in the same phenomenon as was observed by Chatterjee and Pal~\cite{chapalPTRF} (see~\S\ref{sss:intro:capital} above). Yet we provide a different, and complementary, description. In particular, in the case where the market weight of every stock vanishes, we introduce the capital density $\bmu : [0,1] \to [0,+\infty)$ such that the proportion of capital held by the companies ranked between $nu$ and $n(u+\dd u)$ is given by $\bmu(u)\dd u$ in the long-term asymptotic market. The study of the capital density allows us to recover the Pareto-like shape of capital distribution curves, similar to the ones obtained by Fernholz.

We finally address the performance of portfolio rules in Section~\ref{s:portfolio}. We first introduce a family of portfolio rules, called {\em $p$-diversity weighted portfolios}, interpolating between the equally weighted and the market portfolio. The performance of such portfolios is described in terms of the weighted capital measures. Therefore, based on the results of Section~\ref{s:Pi}, we obtain a law of large numbers for the growth rates of these portfolios. Then, we analyse the long time behaviour of these asymptotic growth rates. 

As far as the discussion led in~\S\ref{sss:intro:portfolio} is concerned, we draw the following conclusions: in the limit of a large market, the relative performance of the equally weighted portfolio with respect to the market portfolio only depends on the volatility structure of the market model, and no longer on the growth rate. In particular, if the variance of a stock is a nonincreasing function of its capitalization, which matches the volatility size effect described in~\S\ref{sss:intro:lt}, then we recover the fact that the equally weighted portfolio beats the market. However, we also provide an example of a model, where large stocks have large variance, in which the market portfolio outperforms the equally weighted portfolio, in spite of rebalancing.


\section{The mean-field Atlas model}\label{s:meanfield}

In this section, we give a general description of the limit of the mean-field Atlas model when the number of companies $n$ grows to infinity, laying particular emphasis on the long time behaviour of the market. Our analysis is based on the theoretical study~\cite{jourey}, the main results of which shall be recalled whenever needed. Notations and conventions are set up in Subsection~\ref{ss:prelim}. The description of the large market asymptotics is made in Subsection~\ref{ss:chaos}, and its long time behaviour is discussed in Subsection~\ref{ss:equimarket}.

\subsection{Preliminaries}\label{ss:prelim} We first set up some notations and conventions.

\subsubsection{Assumptions} Let us introduce and discuss the various assumptions that we shall use on the initial distribution $m$ and the coefficients $\gamma$ and $\sigma$ of the mean-field Atlas model.

Following~\cite{bass}, a sufficient condition for the system~\eqref{eq:model} to be defined in the mean-field Atlas model is the following {\em uniform ellipticity} assumption
\begin{equation}\label{hyp:UE}
  \forall u \in [0,1], \qquad \sigma^2(u)>0.
  \tag{UE}
\end{equation}
A weakening of this assumption, allowing degeneracies in $0$ and $1$, is discussed in~\cite[Chapitres 5 et 6]{these}.

The law of large numbers for the weighted capital measure requires integrability conditions on the powers of the capitalization processes. These conditions are propagated from integrability conditions on the powers of initial capitalizations, therefore we shall assume that the common probability distribution $m$ of the initial log-capitalizations $Y^1_n(0), \ldots, Y^n_n(0)$ satisfies
\begin{equation}\label{hyp:H}
  \forall p \geq 0, \qquad \int_{y \in \R} \ex{py} m(\dd y) < +\infty.
  \tag{H}
\end{equation}

We now define the function $\Gamma$ on $[0,1]$ by, for all $u \in [0,1]$,
\begin{equation*}
  \Gamma(u) := \int_{v=0}^u \gamma(v)\dd v.
\end{equation*}
Then, the long-term stability of large markets is ensured by the following equilibrium assumptions~\eqref{hyp:E1} and~\eqref{hyp:E2}. The first one is the continuous equivalent of~\eqref{eq:stabAtlas}, namely
\begin{equation}\label{hyp:E1}
  \forall u \in (0,1), \qquad \Gamma(u) - gu > 0,
  \tag{E1}
\end{equation}
where we recall that $g$ is the market mean growth rate defined in~\eqref{def:g}. Note that~\eqref{hyp:E1} is the continuous translation of the size effect, and that combining this assumption with the continuity of $\gamma$ implies that $\gamma(0) \geq g \geq \gamma(1)$: in average, small stocks grow faster than the market, while large stocks grow slower than the market. In particular, if the growth rate function $\gamma$ is decreasing on $[0,1]$, then Assumption~\eqref{hyp:E1} is satisfied.

The second equilibrium condition writes
\begin{equation}\label{hyp:E2}
  \int_{u=0}^{1/2} \frac{u}{|\Gamma(u)-gu|} \dd u + \int_{u=1/2}^1 \frac{1-u}{|\Gamma(u)-gu|}\dd u < +\infty,
  \tag{E2}
\end{equation}
and ensures integrability properties for the equilibrium distribution. Note that under Assumption~\eqref{hyp:E1} and because of the continuity of $\gamma$, a sufficient condition for~\eqref{hyp:E2} to hold is $\gamma(0) > g > \gamma(1)$, which is a stronger expression of the size effect.

Let us finally note that the growth rate function corresponding to the mean-field approximation of the Atlas model introduced in Remark~\ref{rk:atlas} satisfies the equilibrium conditions~\eqref{hyp:E1} and~\eqref{hyp:E2} for all $\alpha > 0$.

\subsubsection{Notations} For all $T>0$, the space of continuous sample-paths $C([0,T],\R)$ is endowed with the sup norm $||\cdot||_{\infty}$, and the space $C([0,+\infty),\R)$ is provided with the topology of the locally uniform convergence. For all $k \geq 1$, the set of probability distributions on $C([0,+\infty),\R^k)$ is denoted by $\P(C([0,+\infty),\R^k))$ and endowed with the topology of weak convergence. The marginal distribution of $P \in \P(C([0,+\infty),\R))$ at time $t \geq 0$ is denoted by $P_t$. The cumulative distribution function of $P_t$ is denoted by $F_t := H*P_t$, where $H*\cdot$ refers to the convolution with the Heaviside function $H(y) := \ind{y \geq 0}$. For all nonincreasing function $a : \R \to \R$, the pseudo-inverse of $a$ is defined by $a^{-1}(u) := \inf\{y \in \R : a(y) > u\}$.

For all $q \in [1,+\infty)$, the $q$-Wasserstein distance between two cumulative distribution functions $F$ and $G$ on $\R$ is defined by
\begin{equation}\label{eq:coupl}
  W_q(F,G) := \inf_{(X,Y) \in \mathrm{Coupl}(F,G)} \left(\Exp(|X-Y|^q)\right)^{1/q},
\end{equation}
where $\mathrm{Coupl}(F,G)$ refers to the set of random pairs $(X,Y)$ with marginal cumulative distribution functions $F$ and $G$, see Rachev and Rüschendorf~\cite{rachrusch}. The right-hand side above can actually be rewritten in terms of the pseudo-inverse functions $F^{-1}$ and $G^{-1}$ as follows: given a uniform random variable $U$ on $[0,1]$, an optimal coupling is provided by the random pair $(F^{-1}(U),G^{-1}(U)) \in \mathrm{Coupl}(F,G)$~\cite[Theorem~3.1.2, p.~109]{rachrusch}, so that
\begin{equation}\label{eq:wasserstein}
  W_q(F,G) = \left(\int_{u=0}^1 |F^{-1}(u)-G^{-1}(u)|^q\dd u\right)^{1/q}.
\end{equation}

Finally, if $\Pi$ refers to a probability distribution on $[0,1]$, for all measurable and bounded function $f : [0,1] \to \R$, we denote
\begin{equation*}
  \crochet{f}{\Pi} := \int_{u=0}^1 f(u) \Pi(\dd u).
\end{equation*}

\subsection{Propagation of chaos and nonlinear log-capitalization process}\label{ss:chaos} We first recall the following propagation of chaos result from~\cite[Corollary~2.13]{jourey}. For an introduction to the propagation of chaos phenomenon, we refer to the lecture notes by Sznitman~\cite{sznitman}.

\begin{theo}\label{theo:chaos}
  Let us assume that the variance function $\sigma^2$ satisfies the uniform ellipticity condition~\eqref{hyp:UE}, and that the probability distribution $m$ admits a finite first order moment. Recall that $Y^1_n(0), \ldots, Y^n_n(0)$ are i.i.d. according to $m$.
  \begin{itemize}
    \item There exists a unique weak solution $(Y(t))_{t \geq 0}$ to the stochastic differential equation, nonlinear in the sense of McKean,
    \begin{equation}\label{eq:nlSDE}
      \left\{\begin{aligned}
        & \dd Y(t) = \gamma(F_t(Y(t)))\dd t + \sigma(F_t(Y(t))) \dd B(t),\\
        & \text{$F_t=H*P_t$ is the cumulative distribution function of $Y(t)$},
      \end{aligned}\right.
    \end{equation}
    where $Y(0)$ is distributed according to $m$ and $(B(t))_{t \geq 0}$ is a standard brownian motion in $\R$ independent of $Y(0)$. Let $P \in \P(C([0,+\infty),\R))$ denote the law of $(Y(t))_{t \geq 0}$.
    \item For any finite set $\{i_1, \ldots, i_k\}$ of distinct indices, the joint law of $(Y_n^{i_1}(t), \ldots, Y_n^{i_k}(t))_{t \geq 0}$ converges weakly, in $\P(C([0,+\infty),\R^k))$, to the law $P^{\otimes k}$ of $k$ independent copies of the process $(Y(t))_{t \geq 0}$.
    \item Finally, $\dd t$-almost everywhere, the probability distribution $P_t$ is absolutely continuous with respect to the Lebesgue measure on $\R$.
  \end{itemize}
\end{theo}

Nonlinearity in the sense of McKean has to be understood as the fact that the coefficients in the stochastic differential equation~\eqref{eq:nlSDE} depend on the entire law of the random variable $Y(t)$ through its cumulative distribution function $F_t$. Therefore, the process $(Y(t))_{t \geq 0}$ shall be called the {\em nonlinear log-capitalization process}. 

\begin{rk}\label{rk:chaos}
  The following remarks on the nonlinear log-capitalization process can be formulated.
  \begin{itemize}
    \item The equality~\eqref{eq:nlSDE} rewrites
    \begin{equation*}
      Y(t) = Y(0) + \int_{s=0}^t \gamma(F_s(Y(s)))\dd s + \int_{s=0}^t \sigma(F_s(Y(s)))\dd B(s).
    \end{equation*}
    On the one hand, since $\sigma$ is bounded, then the stochastic integral is a centered martingale. On the other hand, by Theorem~\ref{theo:chaos}, $\dd s$-almost everywhere, the probability distribution $P_s$ does not weight points so that $F_s(Y(s))$ has a uniform distribution on $[0,1]$. As a consequence, taking the expectation of the equality above yields~\eqref{eq:EY}, \ie $\Exp(Y(t))=\Exp(Y(0))+gt$.
        
    \item Theorem~\ref{theo:chaos} rewrites as a law of large numbers for the empirical distribution $\nu^n$ of the process of log-capitalizations $(Y^1_n(t), \ldots, Y^n_n(t))_{t \geq 0}$, defined as a random variable in $\P(C([0,+\infty),\R))$ by
\begin{equation*}
  \nu^n := \frac{1}{n} \sum_{i=1}^n \delta_{(Y^i_n(t))_{t \geq 0}}.
\end{equation*}
  Indeed, following~\cite[Proposition~2.2, p.~177]{sznitman}, the statement of Theorem~\ref{theo:chaos} combined with the exchangeability of the processes $(Y^1_n(t))_{t \geq 0}, \ldots, (Y^n_n(t))_{t \geq 0}$ imply that $\nu^n$ converges in probability, in $\P(C([0,+\infty),\R))$, to $P$.
  \end{itemize}
\end{rk}

\subsection{Long time behaviour of the nonlinear log-capitalization process}\label{ss:equimarket} We now describe the long time behaviour of the nonlinear log-capitalization process $(Y(t))_{t \geq 0}$. Because of~\eqref{eq:EY}, it is necessary to introduce a shift by defining $\tY(t) := Y(t) - gt$, for all $t \geq 0$. The process $(\tY(t))_{t \geq 0}$ is called the {\em fluctuation} process, and we denote by $\tF_t(y) := \Pr(\tY(t) \leq y)$ its marginal cumulative distribution function, so that
\begin{equation}\label{eq:tF}
  \tF_t(y) = \Pr(Y(t)-gt \leq y) = F_t(y+gt).
\end{equation}
We first note that $(\tY(t))_{t \geq 0}$ satisfies the same nonlinear stochastic differential equation~\eqref{eq:nlSDE} as $(Y(t))_{t \geq 0}$, with shifted growth rate $\tilde{\gamma}(u) := \gamma(u)-g$.

\begin{lem}
  Under the assumptions of Theorem~\ref{theo:chaos}, the fluctuation $(\tY(t))_{t \geq 0}$ solves the nonlinear stochastic differential equation 
  \begin{equation*}
    \left\{\begin{aligned}
      & \dd \tY(t) = \tilde{\gamma}(\tF_t(\tY(t)))\dd t + \sigma(\tF_t(\tY(t))) \dd B(t),\\
      & \text{$\tF_t$ is the cumulative distribution function of $\tY(t)$},
    \end{aligned}\right.
  \end{equation*}
  where $\tY(0)$ is distributed according to $m$ and $(B(t))_{t \geq 0}$ is a standard brownian motion in $\R$ independent of $\tY(0)$. Moreover, weak uniqueness holds for this nonlinear stochastic differential equation.
\end{lem}
\begin{proof}
  Weak uniqueness for the nonlinear stochastic differential equation follows from the application of Theorem~\ref{theo:chaos} with $\gamma$ replaced with $\tilde{\gamma}$, and it is a straightforward consequence of~\eqref{eq:tF} that $(\tY(t))_{t \geq 0}$ is a solution.
\end{proof}

To describe the long time behaviour of the fluctuation $(\tY(t))_{t \geq 0}$, we now assume that the uniform ellipticity condition~\eqref{hyp:UE} and the equilibrium condition~\eqref{hyp:E1} hold. This enables us to define the continuous, increasing function $\Psi$ on $(0,1)$ by
\begin{equation*}
  \forall u \in (0,1), \qquad \Psi(u) := \int_{v=1/2}^u \frac{\sigma^2(v)}{2(\Gamma(v)-gv)} \dd v.
\end{equation*}
Note that the pseudo-inverse function $\Psi^{-1}$ is a cumulative distribution function on $\R$. Its first order moment writes
\begin{equation*}
  \int_{u=0}^1 \Psi(u)\dd u = -\int_{u=0}^{1/2} u \frac{\sigma^2(u)}{2(\Gamma(u)-gu)}\dd u + \int_{u=1/2}^1 (1-u) \frac{\sigma^2(u)}{2(\Gamma(u)-gu)}\dd u,
\end{equation*}
and, because of Assumption~\eqref{hyp:UE}, it is finite if and only if Assumption~\eqref{hyp:E2} holds.

The cumulative distribution function $\Psi^{-1}$ is a weak solution to the stationary version of the Cauchy problem satisfied by $F_t$, see~\cite{jourey}. It actually describes the equilibrium of the fluctuation process, which is made precise in the following theorem from~\cite[Section~4]{jourey}.

\begin{theo}\label{theo:equi}
  Let us assume that:
  \begin{itemize}
    \item the function $\sigma^2$ satisfies the uniform ellipticity condition~\eqref{hyp:UE},
    \item the function $\gamma$ satisfies the equilibrium conditions~\eqref{hyp:E1} and~\eqref{hyp:E2},
    \item the function $\sigma^2$ is $C^2$ on $[0,1]$, the function $\gamma$ is $C^1$ on $[0,1]$ and there exists $\beta > 0$ such that the functions $(\sigma^2)''$ and $\gamma'$ are $\beta$-Hölder continuous,
    \item the probability distribution $m$ has a finite first order moment, and, for all $p \geq 0$, $W_p(H*m, \Psi^{-1}) < +\infty$.
  \end{itemize}
  Then, the pseudo-inverse $\Psi^{-1}$ of the function $\Psi$ introduced above is the cumulative distribution function of a probability distribution with positive density on $\R$. 
  
  Let us now write $\tF_t$ for the cumulative distribution function of $\tY(t)$, and define $\tF_{\infty}$ by $\tF_{\infty}(y) = \Psi^{-1}(y+\bar{y})$, where $\bar{y}$ is chosen so that 
  \begin{equation*}
    \int_{y \in \R} y \tF'_{\infty}(y)\dd y = \int_{y \in \R} y m(\dd y).
  \end{equation*} 
  
  Then, for all $p \geq 1$,
  \begin{equation*}
    \lim_{t \to +\infty} W_p(\tF_t, \tF_{\infty}) = 0.
  \end{equation*}
\end{theo}

The probability distribution with density $(\Psi^{-1})'$ shall be referred to as the {\em equilibrium distribution}. We discuss the shape of its tails in the following remark.

\begin{rk}\label{rk:tail}
  Describing the tail of the equilibrium distribution in $+\infty$ amounts to describing the behaviour of $\Psi(u)$ when $u \uto 1$. Let us recall that, under Assumptions~\eqref{hyp:UE}, \eqref{hyp:E1} and~\eqref{hyp:E2}, $\gamma(1) \leq g$; so that the {\em critical diversity index} $\pc$ defined by
  \begin{equation}\label{eq:pc}
    \pc := \frac{2(g-\gamma(1))}{\sigma^2(1)}
  \end{equation}
  is nonnegative.
  \begin{itemize}
    \item If $\gamma(1) < g$, that is to say $\pc > 0$, then writing
    \begin{equation*}
      \Gamma(v) - gv = g(1-v) - \int_{w=v}^1 \gamma(w)\dd w
    \end{equation*}
    yields
    \begin{equation*}
      \Psi(u) = \int_{v=1/2}^u \frac{\sigma^2(v)}{2(\Gamma(v) - gv)}\dd v \Sim_{u \uto 1} \frac{\sigma^2(1)}{2(g-\gamma(1))} \int_{v=1/2}^u \frac{\dd v}{1-v} \Sim_{u \uto 1} -\frac{1}{\pc}\log(1-u), 
    \end{equation*}
    so that the tail of the equilibrium distribution in $+\infty$ is expected to be exponential with parameter $\pc$, that is to say, $1 - \Psi^{-1}(y)$ is expected to decay to $0$ at an exponential rate of order $\pc$.
    \item If $\gamma(1)=g$, that is to say $\pc=0$, then the tail of the equilibrium distribution in $+\infty$ is expected to be heavy, that is to say, $1 - \Psi^{-1}(y)$ is expected to decay to $0$ slower than any exponential rate.
  \end{itemize}
  
  Likewise, a symmetric phenomenon is observed for the tail of the equilibrium distribution in $-\infty$. The critical index $\qc$ defined by $\qc := 2(\gamma(0)-g)/\sigma^2(0)$ is nonnegative, and if $\qc > 0$, then 
  \begin{equation*}
    \Psi(u) \Sim_{u \dto 0} \frac{1}{\qc}\log(u),
  \end{equation*}
  so that, when $y \to -\infty$, $\Psi^{-1}(y)$ is expected to decay to $0$ at an exponential rate of order $\qc$. If $\qc = 0$, then the tail of the equilibrium distribution in $-\infty$ is expected to be heavy.
\end{rk}


\section{The weighted capital measure}\label{s:Pi}

For all $p \geq 0$, $t \geq 0$, the weighted capital measure $\Pi^p_n(t)$ is defined by~\eqref{eq:Pin}. For all measurable and bounded function $f : [0,1] \to \R$,
\begin{equation*}
  \crochet{f}{\Pi_n^p(t)} = \sum_{j=1}^n \frac{(X^{(j)}_n(t))^p}{(X^1_n(t))^p + \cdots + (X^n_n(t))^p}f\left(\frac{j}{n}\right) = \sum_{j=1}^n \frac{\ex{pY^{(j)}_n(t)}}{\ex{pY^1_n(t)} + \cdots + \ex{pY^n_n(t)}}f\left(\frac{j}{n}\right).
\end{equation*}

As is explained in Section~\ref{s:intro}, the capital measure is strongly related to the capital distribution curves. Likewise, we shall describe in Section~\ref{s:portfolio} below the link between the weighted capital measures and the performance of a family of portfolio rules. Therefore, it is of interest to describe the asymptotic behaviour of the weighted capital measure when the size of the market grows to infinity. This task is carried out in Subsection~\ref{ss:Pi:lln} by deriving a law of large numbers for $\Pi^p_n(t)$. The corresponding limit $\Pi^p(t)$ is referred to as the asymptotic weighted capital measure, and its long time behaviour is addressed in Subsection~\ref{ss:Pi:lt}. We discuss the validity of the interversion of the limits `$t \to +\infty$' and `$n \to +\infty$' in Subsection~\ref{ss:interv}.

\subsection{Law of large numbers}\label{ss:Pi:lln} We first address the limit, when $n$ grows to infinity, of $\Pi^p_n(t)$.

\begin{prop}\label{prop:lln:Pi}
  Let us assume that the conditions of Theorem~\ref{theo:chaos} are satisfied, and that the probability distribution $m$ satisfies the condition~\eqref{hyp:H}. Let us fix $T>0$ and $q \in [1,+\infty)$. Then, for all $p \geq 0$,
  \begin{itemize}
    \item there exists $C_T^p < +\infty$ such that
    \begin{equation}\label{eq:supexpY}
      \forall t \in [0,T], \qquad \Z^p(t) := \int_{u=0}^1 \ex{p F_t^{-1}(u)}\dd u = \Exp\left(\ex{p Y(t)}\right) \leq C_T^p,
    \end{equation}
    
    \item for all continuous function $f : [0,1] \to \R$, the process $(\scrochet{f}{\Pi^p_n(t)})_{t \in [0,T]}$ converges, in $L^q(C([0,T],\R))$, to the deterministic process $(\scrochet{f}{\Pi^p(t)})_{t \in [0,T]}$, where $\Pi^p(t)$ is the probability distribution with density $\exp(p F_t^{-1}(u)) / \Z^p(t)$ with respect to the Lebesgue measure on $[0,1]$.
  \end{itemize}
\end{prop}

The proof of Proposition~\ref{prop:lln:Pi} is detailed in Appendix~\ref{app:lln}. The probability distribution $\Pi^p(t)$ shall be called the {\em asymptotic weighted capital measure}.

\subsection{Long-term asymptotic capital measure}\label{ss:Pi:lt} We now address the long time behaviour of the asymptotic weighted capital measure $\Pi^p(t)$.

\subsubsection{Heuristic derivation}\label{sss:heuristic} Let us recall that the cumulative distribution function $\tF_t$ of the fluctuation $\tY(t) = Y(t) - gt$ writes $\tF_t(y) = F_t(y+gt)$. As a consequence, the density of the asymptotic weighted capital measure $\Pi^p(t)$ with respect to the Lebesgue measure on $[0,1]$ rewrites
\begin{equation*}
  \frac{\ex{p F_t^{-1}(u)}}{\displaystyle\int_{u=0}^1 \ex{p F_t^{-1}(u)}\dd u} = \frac{\ex{p(\tF_t^{-1}(u)+gt)}}{\displaystyle\int_{u=0}^1 \ex{p(\tF_t^{-1}(u)+gt)}\dd u} = \frac{\ex{p \tF_t^{-1}(u)}}{\displaystyle\int_{u=0}^1 \ex{p \tF_t^{-1}(u)}\dd u}.
\end{equation*}
Under appropriate assumptions, Theorem~\ref{theo:equi} asserts that $\tF_t$ converges, in Wasserstein distance, to $\tF_{\infty}$ defined by $\tF_{\infty}(y) = \Psi^{-1}(y+\bar{y})$, where $\bar{y}$ is chosen so that $\tF_{\infty}$ and $m$ have the same expectation. As a consequence, the asymptotic weighted capital measure $\Pi^p(t)$ is expected to converge to the probability distribution $\bPi^p$ with density
\begin{equation*}
  \frac{\ex{p \tF_{\infty}^{-1}(u)}}{\displaystyle\int_{u=0}^1 \ex{p \tF_{\infty}^{-1}(u)}\dd u} = \frac{\ex{p(\Psi(u)-\bar{y})}}{\displaystyle\int_{u=0}^1 \ex{p(\Psi(u)-\bar{y})}\dd u} = \frac{\ex{p\Psi(u)}}{\displaystyle\int_{u=0}^1 \ex{p\Psi(u)}\dd u},
\end{equation*}
as long as 
\begin{equation*}
  \bZ^p := \int_{u=0}^1 \ex{p\Psi(u)}\dd u < +\infty. 
\end{equation*}
Following the first-order analysis of the equilibrium distribution carried out in Remark~\ref{rk:tail}, this should be the case for $p \in [0, \pc)$. On the contrary, if $p > \pc$, then $\bZ^p$ is expected to be infinite, and all the mass of $\Pi^p(t)$ should concentrates around $1$ when $t$ grows to infinity, so that $\Pi^p(t)$ is rather expected to converge to the Dirac distribution $\delta_1$. This {\em phase transition phenomenon} is made precise in~\S\ref{sss:phtrans} below.

\subsubsection{Phase transition}\label{sss:phtrans} Let us recall that the critical diversity index $\pc \geq 0$ was defined in~\eqref{eq:pc}.

\begin{lem}\label{lem:bPip}
  Let us assume that the uniform ellipticity condition~\eqref{hyp:UE}, that the equilibrium condition~\eqref{hyp:E1} hold, and that the critical diversity index $\pc$ is positive. Then, for all $p \in [0,\pc)$, $\bZ^p < +\infty$, and we denote by $\bPi^p$ the probability distribution with density $\exp(p \Psi(u))/\bZ^p$ with respect to the Lebesgue measure on $[0,1]$.
  
  Moreover, for all continuous function $f : [0,1] \to \R$, the function $p \mapsto \scrochet{f}{\bPi^p}$ is continuous on $[0,\pc)$, and:
  \begin{itemize}
    \item if $\bZ^{\pc} = +\infty$, then $\lim_{p \uto \pc} \scrochet{f}{\bPi^p} = f(1)$,
    \item if $\bZ^{\pc} < +\infty$, we denote by $\bPi^{\pc}$ the probability distribution with density $\exp(\pc \Psi(u))/\bZ^{\pc}$ with respect to the Lebesgue measure on $[0,1]$, and then $\lim_{p \uto \pc} \scrochet{f}{\bPi^p} = \scrochet{f}{\bPi^{\pc}}$.
  \end{itemize}
\end{lem}

The proof of Lemma~\ref{lem:bPip} is postponed to Appendix~\ref{app:longtime}. The probability distribution $\bPi^p$ shall be called the {\em long-term asymptotic weighted capital measure}.

\begin{expl}\label{expl:beta}
  We explicit the long-term asymptotic weighted capital measure for a constant variance function $\sigma^2$ and for $\gamma(u) = 1-2u$. For these coefficients, $g=0$ and the equilibrium distribution was computed in~\cite[Example~2.3]{jm}. In particular, the function $\Psi$ writes
  \begin{equation*}
    \Psi(u) = \int_{v=1/2}^u \frac{\sigma^2}{2v(1-v)}\dd v = \frac{1}{\pc}\log\left(\frac{u}{1-u}\right),
  \end{equation*}
  so that, for $p \in [0,\pc)$, $\bPi^p$ is the $\Beta(1+p/\pc,1-p/\pc)$ distribution. In addition, it is easily checked that $\bZ^{\pc} = +\infty$, so that $\bPi^p$ converges to the Dirac distribution in $1$ when $p$ approaches the critical diversity index $\pc$.

  \begin{figure}[ht]
    \includegraphics[width=15cm]{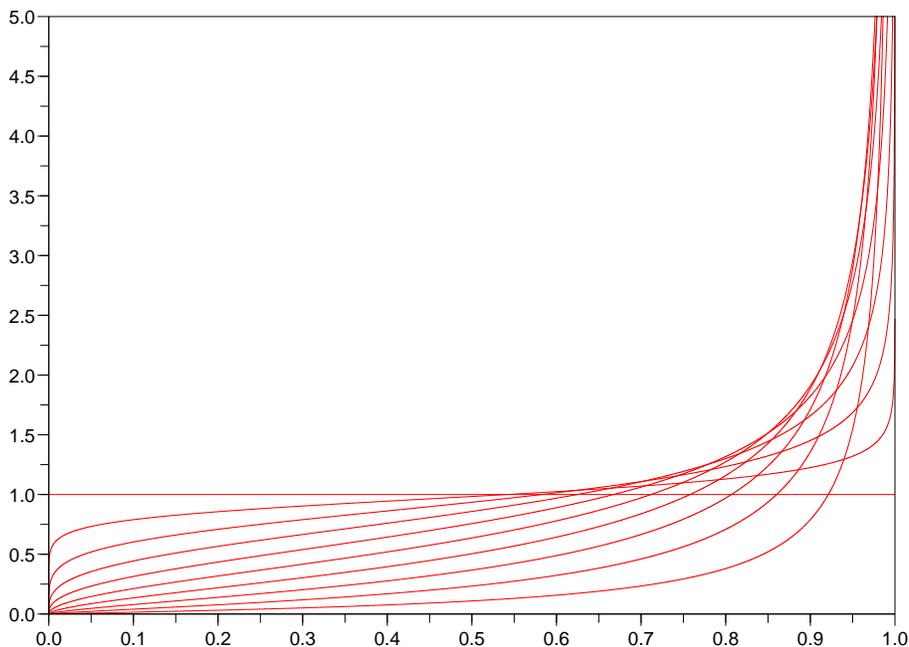}
    \caption{The density of $\bPi^p$ for a constant variance function $\sigma^2$ and $\gamma(u) = 1-2u$. The diversity index $p$ varies between $0$ and $\pc$. The uniform density is recovered for $p=0$, while the mass concentrates on the point $u=1$ when $p$ approaches the critical diversity index $\pc$.}
    \label{fig:beta}
  \end{figure}
\end{expl}

We now explicit the link between $\bPi^p$ and the long time behaviour of $\Pi^p(t)$. 

\begin{prop}\label{prop:phtrans}
  Let us assume that the conditions of Theorem~\ref{theo:equi} hold, and that the probability distribution $m$ satisfies the condition~\eqref{hyp:H}. Let $\pc \geq 0$ be defined by~\eqref{eq:pc}. Let $f : [0,1] \to \R$ be a continuous function, and $p \geq 0$.
  \begin{itemize}
    \item {\em Subcritical phase:} if $p \in [0,\pc)$, then
    \begin{equation*}
      \lim_{t \to +\infty} \crochet{f}{\Pi^p(t)} = \crochet{f}{\bPi^p},
    \end{equation*}
    where the probability distribution $\bPi^p$ is defined in Lemma~\ref{lem:bPip}.
    
    \item {\em Supercritical phase:} if $p > \pc$, then
    \begin{equation*}
      \lim_{t \to +\infty} \crochet{f}{\Pi^p(t)} = f(1).
    \end{equation*}
  
    \item {\em Criticality:} the long time behaviour of $\scrochet{f}{\Pi^{\pc}(t)}$ is described as follows:
    \begin{itemize}
      \item if $\bZ^{\pc} = +\infty$, by Lemma~\ref{lem:bPip}, $\lim_{p \uto \pc} \scrochet{f}{\bPi^p}=f(1)$, and then 
      \begin{equation*}
        \lim_{t \to +\infty} \crochet{f}{\Pi^{\pc}(t)} = f(1),
      \end{equation*}
      
      \item if $\bZ^{\pc} < +\infty$, by Lemma~\ref{lem:bPip}, $\lim_{p \uto \pc} \scrochet{f}{\bPi^p} = \scrochet{f}{\bPi^{\pc}}$, and then
      \begin{equation}\label{eq:encadre}
        f(1) \wedge \crochet{f}{\bPi^{\pc}} \leq \liminf_{t \to +\infty} \crochet{f}{\Pi^{\pc}(t)} \leq \limsup_{t \to +\infty} \crochet{f}{\Pi^{\pc}(t)} \leq f(1) \vee \crochet{f}{\bPi^{\pc}}.
      \end{equation} 
    \end{itemize}
  \end{itemize}  
\end{prop}

The proof of Proposition~\ref{prop:phtrans} is postponed to Appendix~\ref{app:longtime}. The description of the long time behaviour of $\Pi^p(t)$ is summarized on Figure~\ref{fig:limt}. Note that, in the case $\bZ^{\pc} = +\infty$, the function $p \mapsto \lim_{t \to +\infty} \scrochet{f}{\Pi^p(t)}$ is defined and continuous on $[0,+\infty)$.

\begin{figure}[ht]
  \includegraphics[width=15cm]{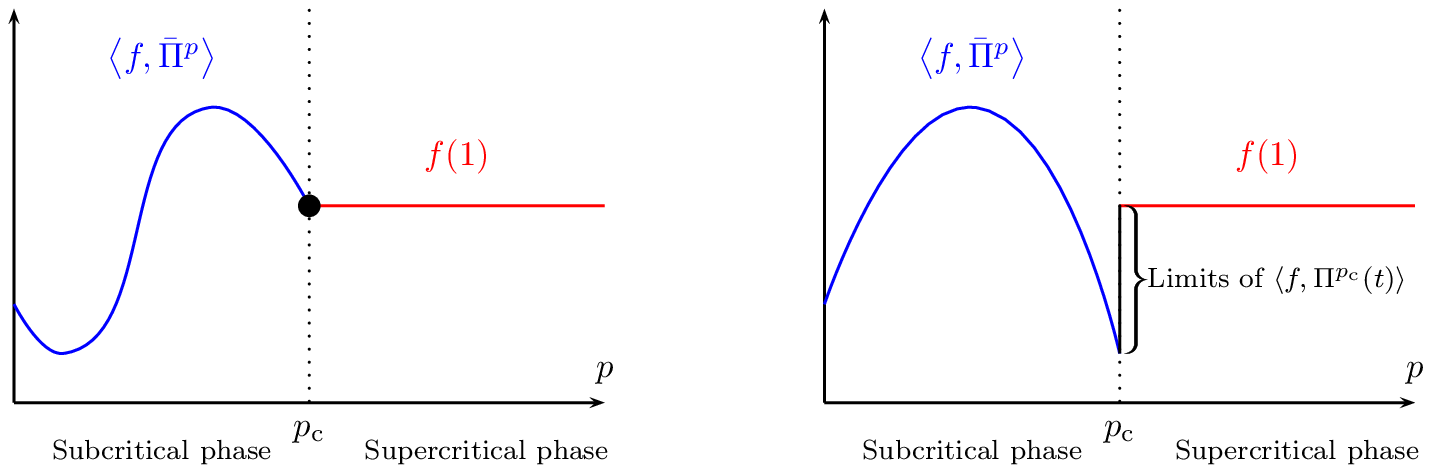}
  \caption{A schematic representation of the long time behaviour of $\scrochet{f}{\Pi^p(t)}$ according to Proposition~\ref{prop:phtrans}. On the left-hand figure, $\lim_{p \uto \pc} \scrochet{f}{\bPi^p} = f(1)$, so that $\lim_{t \to +\infty} \scrochet{f}{\Pi^{\pc}(t)}$ exists and its value is represented by the black dot. On the right-hand figure, $\lim_{p \uto \pc} \scrochet{f}{\bPi^p} \not= f(1)$ and the limit points of $\scrochet{f}{\Pi^{\pc}(t)}$ are located inside the gap between $\lim_{p \uto \pc} \scrochet{f}{\bPi^p} = \scrochet{f}{\bPi^{\pc}}$ and $f(1)$.}
  \label{fig:limt}
\end{figure}

\subsection{Interversion of the limits}\label{ss:interv} The results of Subsections~\ref{ss:Pi:lln} and~\ref{ss:Pi:lt} formally rewrite as follows:
\begin{equation*}
  \lim_{t \to +\infty} \lim_{n \to +\infty} \Pi^p_n(t) = \begin{cases}
    \bPi^p & \text{if $p \in [0, \pc)$,}\\
    \delta_1 & \text{if $p>\pc$.}
  \end{cases}
\end{equation*}
In order to compare our results with those of~\cite{banner, chapalPTRF}, it would be desirable to check that it holds
\begin{equation*}
  \lim_{n \to +\infty} \lim_{t \to +\infty} \Pi^p_n(t) = \begin{cases}
    \bPi^p & \text{if $p \in [0, \pc)$,}\\
    \delta_1 & \text{if $p>\pc$.}
  \end{cases}
\end{equation*}
In this subsection, we outline the arguments allowing to obtain such an interversion property. We shall restrict our study to the case of a constant variance function $\sigma^2$, for which the long time behaviour of the weighted capital measure, for a fixed size $n$ of the market, can be explicitely described thanks to the results of Pal and Pitman~\cite{pp} and Jourdain and Malrieu~\cite{jm}.

For $n \geq 1$, we denote by $(\bar{Y}^1_n(t), \ldots, \bar{Y}^n_n(t))$ the vector defined by, for all $i \in \{1, \ldots, n\}$,
\begin{equation*}
  \bar{Y}^i_n(t) := Y^i_n(t) - \frac{1}{n} \sum_{j=1}^n Y^j_n(t).
\end{equation*}
Then, under Assumption~\eqref{eq:stabAtlas}, which is the discrete version of~\eqref{hyp:E1}, Theorem~8 in~\cite{pp} asserts that $(\bar{Y}^1_n(t), \ldots, \bar{Y}^n_n(t))$ converges in distribution, when $t$ grows to infinity, to a random vector $(\bar{\xi}^1_n, \ldots, \bar{\xi}^n_n)$ with explicit symmetric law $P_n$. This convergence statement can be expressed in total variation~\cite{pp} or $\chi_2$ distance~\cite{jm}. As a consequence, for all continuous function $f : [0,1] \to \R$, for all $p \geq 0$,
\begin{equation*}
  \crochet{f}{\Pi^p_n(t)} = \sum_{j=1}^n \frac{\ex{p\bar{Y}^{(j)}_n(t)}}{\ex{p\bar{Y}^1_n(t)} + \cdots + \ex{p\bar{Y}^n_n(t)}}f\left(\frac{j}{n}\right)
\end{equation*}
depends continuously on $(\bar{Y}^1_n(t), \ldots, \bar{Y}^n_n(t))$, and therefore it converges in distribution, when $t$ grows to infinity, to
\begin{equation*}
  \sum_{j=1}^n \frac{\ex{p\bar{\xi}^{(j)}_n}}{\ex{p\bar{\xi}^1_n} + \cdots + \ex{p\bar{\xi}^n_n}}f\left(\frac{j}{n}\right) =: \crochet{f}{\bPi^p_n},
\end{equation*}
where $(\bar{\xi}^{(1)}_n, \ldots, \bar{\xi}^{(n)}_n)$ obviously refers to the increasing reordering of $(\bar{\xi}^1_n, \ldots, \bar{\xi}^n_n)$.

Following~\cite{chaoticity}, the stationary distribution $P_n$ of the process $(\bar{Y}^1_n(t), \ldots, \bar{Y}^n_n(t))_{t \geq 0}$ is chaotic with respect to the centered equilibrium distribution of the fluctuation process, given by Theorem~\ref{theo:equi}, and that we shall denote by $\bar{P}$. On account of the symmetry of $P_n$, Proposition~2.2 in~\cite{sznitman} allows us to rewrite the result of~\cite{chaoticity} as a law of large numbers for the (random) empirical distribution $\bar{\nu}^n$ of $(\bar{\xi}^1_n, \ldots, \bar{\xi}^n_n)$; namely,
\begin{equation*}
  \bar{\nu}^n := \frac{1}{n} \sum_{i=1}^n \delta_{\bar{\xi}^i_n}
\end{equation*}
converges in probability, in the space $\P(\R)$ of probability distributions on the real line, endowed with the topology of weak convergence, toward $\bar{P}$. Following the lines of the proof of Proposition~\ref{prop:lln:Pi}, but replacing the integrability properties of $\exp(p\tY(t))$ and $\exp(pY^1_n(t))$ with the results from~\cite{chaoticity} regarding the Laplace transforms of $\bar{P}$ and of the first marginal $P_{1,n}$ of $P_n$, we obtain that, in the subcritical case, for all $q \in [1,+\infty)$, the real valued random variable $\scrochet{f}{\bPi^p_n}$ converges in $L^q(\R)$ to 
\begin{equation*}
  \frac{\displaystyle \int_{u=0}^1 f(u) \ex{p (H*\bar{P})^{-1}(u)} \dd u}{\displaystyle \int_{u=0}^1 \ex{p (H*\bar{P})^{-1}(u)} \dd u} = \crochet{f}{\bPi^p}.
\end{equation*}
Similarly, in the supercritical case, arguing as in the proof of Proposition~\ref{prop:phtrans} allows us to obtain that $\scrochet{f}{\bPi^p_n}$ converges to $f(1)$.

As a conclusion, at least for models with a constant variance function, the limit when both $n$ and $t$ grow to infinity of quantities of the form $\scrochet{f}{\Pi^p_n(t)}$, with $p \not= \pc$, does not depend on the order in which the limits are taken.


\section{Capital distribution curves}\label{s:capital}

We pursue the discussion of~\S\ref{sss:intro:results} in order to describe the capital distribution in the long-term asymptotic mean-field Atlas model. Recall that for $p=1$, we write $\Pi_n(t)$ instead of $\Pi_n^1(t)$, and refer to this measure as the capital measure. Similarly, we write $\Pi(t)$ instead of $\Pi^1(t)$ for the asymptotic capital measure, and $\bPi$ instead of $\bPi^1$ for the long-term asymptotic capital measure when it is defined. The associated normalization constants are denoted by $\Z(t)$ and $\bZ$.

If $[k]$ refers to the index of the company with $k$-th largest capitalization, we define the {\em relative rank} of this company by $k/n \in [0,1]$. In the limit of large markets, we shall be interested by the proportion of capital held by companies with relative rank between $u$ and $u + \dd u$, for $u \in [0,1]$.

\subsection{Phase transition for the long-term asymptotic capital measure}\label{ss:capital:phtrans} We first recall the following technical lemma, which is a straighforward consequence of the Portmanteau theorem~\cite[Theorem~2.1, p.~11]{billingsley}.

\begin{lem}\label{lem:portmanteau}
  Let $(\Pi_n)_{n \geq 1}$ be a sequence of probability distributions on $\R$, such that $\Pi_n$ converges weakly to a probability distribution $\Pi$ on $[0,1]$. If $\Pi$ is absolutely continuous with respect to the Lebesgue measure on $[0,1]$, then for all interval $I \subset [0,1]$, $\Pi_n(I)$ converges to $\Pi(I)$.
\end{lem}

We deduce from~\S\ref{sss:intro:results}, Proposition~\ref{prop:lln:Pi} and Lemma~\ref{lem:portmanteau}, that for all $t \geq 0$ and $u,v \in [0,1]$ with $u \leq v$, the proportion of capital held by the companies with relative rank between $u$ and $v$ converges in probability to
\begin{equation*}
  \crochet{\ind{1-v \leq \cdot \leq 1-u}}{\Pi(t)} = \frac{1}{\Z(t)} \int_{w=u}^v \ex{F_t^{-1}(1-w)}\dd w.
\end{equation*} 
In particular, the proportion of capital held by the companies with relative rank between $u$ and $u+\dd u$ in a large market is roughly $\exp(F_t^{-1}(1-u))\dd u/\Z(t)$. Then, the phase transition phenomenon derived in Section~\ref{s:Pi} translates as follows. 
\begin{enumerate}[label=(\roman*), ref=\roman*]
  \item\label{it:capital:sub} If $\pc > 1$, then the asymptotic capital measure (with index $p=1$) is subcritical, so that in the long-term, the proportion of capital held by the companies with relative rank between $u$ and $u + \dd u$ is roughly $\bmu(u)\dd u$, where
  \begin{equation*}
    \bmu(u) := \frac{\ex{\Psi(1-u)}}{\bZ}
  \end{equation*}
  is the {\em capital density}.
  
  \item\label{it:capital:super} If $\pc < 1$, the asymptotic capital measure is supercritical, therefore $\Pi(t)$ converges weakly to the Dirac distribution $\delta_1$. As a consequence, all the capital concentrates on the relative rank $0$.
\end{enumerate} 

A detailed study of the capital density $\bmu$ is carried out in Subsection~\ref{ss:capital:curve}, and the Pareto-like distribution empirically observed is recovered. We establish a comparison between our results and the article by Chatterjee and Pal~\cite{chapalPTRF} in Subsection~\ref{ss:capital:chapal}.

\subsection{Capital distribution curve in the subcritical case}\label{ss:capital:curve} Let us assume that $\pc > 1$. Similarly to Fernholz~\cite[Section~5]{fernholz}, we call {\em capital distribution curve} the logarithmic plot of the function $u \mapsto \bmu(u)$. For the coefficients introduced in Example~\ref{expl:beta}, we draw the capital distribution curve on Figure~\ref{fig:capital}.

\begin{figure}[ht]
  \includegraphics[width=15cm]{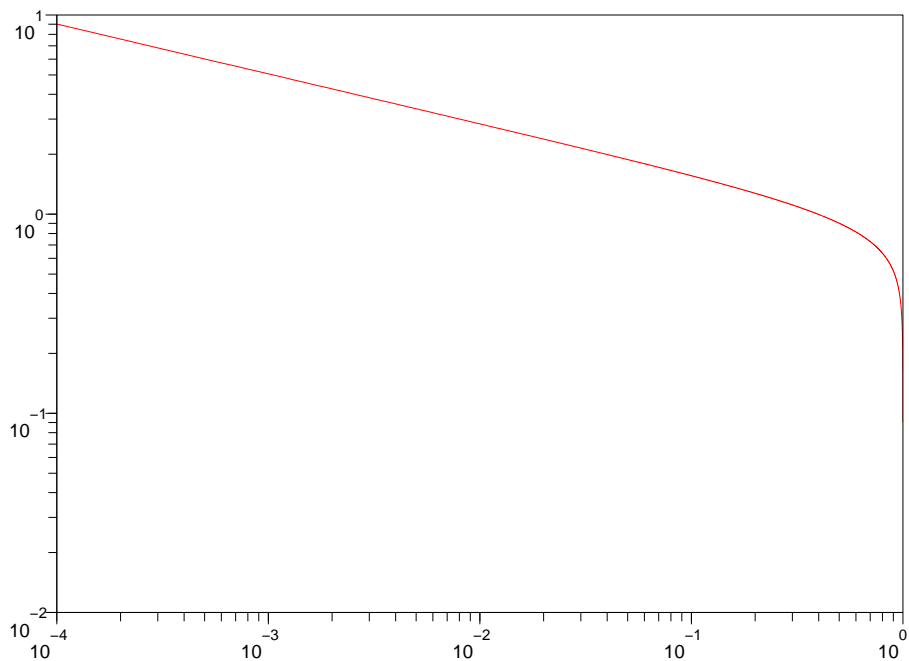}
  \caption{The logarithmic plot of the capital density $\bmu(u)$ for the coefficients of Example~\ref{expl:beta}, with $\pc=4$.}
  \label{fig:capital}
\end{figure}

Figure~\ref{fig:capital} has to be compared with the shape of the empirical curves obtained by Fernholz~\cite[Figure~5.1, p.~95]{fernholz}, which exhibit the following characteristics:
\begin{itemize}
  \item they are almost linear for stocks with small ranks, which indicates a Pareto-like distribution of the capital,
  \item they become concave for stocks with large ranks.
\end{itemize}
This behaviour is easily recovered for the long-term asymptotic capital measure.

\begin{prop}
  Let us assume that the conditions of Proposition~\ref{prop:phtrans} hold, with $\pc > 1$. Then:
  \begin{itemize}
    \item for $u \dto 0$, $\log \bmu(u) \sim (-1/\pc) \log u$, therefore the capital distribution curve is linear with slope $-1/\pc \in (-1, 0)$,
    
    \item for $u \uto 1$, $\log \bmu(u) \to -\infty$ and, if the parameter $\qc$ defined in Remark~\ref{rk:tail} is positive, then $\log \bmu(u) \sim (1/\qc) \log(1-u)$.
  \end{itemize}
\end{prop}
\begin{proof}
  By the definition of the capital density, $\log \bmu(u) = \Psi(1-u) - \log \bZ$. Therefore, the asymptotic behaviour of $\log \bmu(u)$ in $0$ and $1$ is a straightforward consequence of Remark~\ref{rk:tail}.
\end{proof}

\subsection{The Chatterjee-Pal phase transition}\label{ss:capital:chapal} We now describe the phase transition observed by Chatterjee and Pal in~\cite{chapalPTRF}, and discuss the relation with the long time behaviour of our asymptotic capital measure. 

Let us assume that, for all $u \in [0,1]$, $\sigma(u)=1$, and that $\gamma$ is decreasing (so that it satisfies the equilibrium assumption~\eqref{hyp:E1}). Then, following~\cite{pp}, the process of market weights $(\mu^1_n(t), \ldots, \mu^n_n(t))_{t \geq 0}$, defined by
\begin{equation*}
  \mu^i_n(t) = \frac{X^i_n(t)}{X^1_n(t) + \cdots + X^n_n(t)},
\end{equation*}
admits a unique stationary distribution. Let us sample $(\mu^1_n, \ldots, \mu^n_n)$ from this distribution, and denote by $(\mu^{[1]}_n, \ldots, \mu^{[n]}_n)$ the decreasing reordering of $(\mu^1_n, \ldots, \mu^n_n)$. 

The set of sequences $(m_k)_{k \geq 1}$ such that $1 \geq m_1 \geq m_2 \geq \cdots \geq 0$ is endowed with the distance $d(m,m') := \sum_{k=1}^{+\infty} 2^{-k} (|m_k-m'_k| \wedge 1)$. Then, Theorem~2 in~\cite{chapalPTRF} writes as follows: when $n$ grows to infinity,
\begin{itemize}
  \item if $\pc=0$, then the sequence $(\mu^{[1]}_n, \ldots, \mu^{[n]}_n)$ converges in probability to $(1, 0, \ldots)$,
  \item if $\pc \in (0,1)$, then the sequence $(\mu^{[1]}_n, \ldots, \mu^{[n]}_n)$ converges in distribution to a Poisson-Dirichlet process with parameter $\pc$,
  \item if $\pc \geq 1$, then the sequence $(\mu^{[1]}_n, \ldots, \mu^{[n]}_n)$ converges in probability to $(0, 0, \ldots)$.
\end{itemize}
The Poisson-Dirichlet process, introduced by Kingman~\cite{kingman}, is a random sequence $(m_k)_{k \geq 1}$ such that, almost surely, $1 > m_1 > m_2 > \cdots > 0$ and $\sum_{k=1}^{+\infty} m_k = 1$. In particular, $m_k$ converges to $0$ as $k$ increases.

In the case $\pc < 1$, let $(m_k)_{k \geq 1}$ refer to the limit, when $n \to +\infty$, of the sequence $(\mu^{[1]}_n, \ldots, \mu^{[n]}_n)$. It is either $(1, 0, \ldots)$ or a Poisson-Dirichlet process. In both cases, given $u \in (0,1)$, the number of companies with rank lower than $nu$ tends to infinity when $n$ grows. As a consequence, for $n$ large enough, an arbitrarily large proportion of the capital is held by companies with rank lower than $nu$. In the limit, all the mass of the long-term asymptotic capital measure is therefore concentrated below the level $u$. Since $u$ may be arbitrarily small, we obtain that, in the large market limit, the whole capital is held by companies with relative rank around $0$, which is our conclusion of the supercritical case~\eqref{it:capital:super} in Subsection~\ref{ss:capital:phtrans}.

In the case $\pc \geq 1$, all the market weights vanish. This is coherent with~\eqref{it:capital:sub} in Subsection~\ref{ss:capital:phtrans}, since the measure $\bmu(u)\dd u$ does not weight points, so that no company holds a positive proportion of capital when $n$ grows to infinity. However, our study of the capital density $\bmu$ provides informations on the capital distribution that are not available from Chatterjee and Pal's results.

As a conclusion, although we observe the very same phenomenon as Chatterjee and Pal, we depict it differently. In particular, they give detailed informations on the supercritical phase that our study cannot recover, while we provide a more precise description the capital distribution in the subcritical phase.


\section{Performance of diversity weighted portfolios}\label{s:portfolio}

We finally address the analysis of the performance of diversity weighted portfolios. The mathematical framework of Stochastic Portfolio Theory is briefly recalled in Subsection~\ref{ss:spt}, where we also introduce a family of portfolios, called diversity weighted portfolios. This family is indexed by a diversity index and interpolates between the equally weighted portfolio and the market portfolio.

The performance of a portfolio rule is measured by its long-term asymptotic growth rate and excess growth rate, that we define in Subsection~\ref{ss:G}. The monotonicity of these quantities with respect to the diversity index is addressed in Subsection~\ref{ss:monoton}, and a reduction formula providing simple expressions is derived in Subsection~\ref{ss:reduction}. 

We use these results to explicit the long-term asymptotic growth rate of the equally weighted portfolio and the market portfolio in Subsection~\ref{ss:perf}, and state global conclusions in Subsection~\ref{ss:optimal}.

\subsection{Stochastic portfolio theory in a nutshell}\label{ss:spt} We first provide a short overview of Stochastic Portfolio Theory~\cite{fernholz, ferkar}.

\subsubsection{Portfolio} A {\em portfolio rule}, or {\em portfolio} for short, is an adapted process 
\begin{equation*}
  \pi_n = (\pi^1_n(t), \ldots, \pi^n_n(t))_{t \geq 0}
\end{equation*}
such that, for all $t \geq 0$, for all $i \in \{1, \ldots, n\}$, $\pi^i_n(t) \geq 0$ and $\pi^1_n(t) + \cdots + \pi^n_n(t) = 1$. It describes the proportion of wealth that one invests in each stock. We assume that portfolios are self-financing, that is to say, there is no exogenous infusion or withdrawal of money after the initial time. Then, the {\em wealth process} $(Z_n^{\pi_n}(t))_{t \geq 0}$ associated with a portfolio $\pi_n$ satisfies
\begin{equation*}
  \frac{\dd Z_n^{\pi_n}(t)}{Z_n^{\pi_n}(t)} = \sum_{i=1}^n \pi^i_n(t) \frac{\dd X^i_n(t)}{X^i_n(t)},
\end{equation*}
and the initial wealth is normalized to $Z_n^{\pi_n}(0)=1$. By Itô's formula, 
\begin{equation*}
  \dd (\log Z^{\pi_n}_n(t)) = \gamma^{\pi_n}_n(t)\dd t + \sum_{i=1}^n \pi^i_n(t)\sigma^i_n(t)\dd B^i(t),
\end{equation*}
where the processes
\begin{equation*}
  \gamma^{\pi_n}_n(t) := \sum_{i=1}^n \pi^i_n(t)\gamma^i_n(t) + \gamma^{\pi_n}_{*,n}(t), \qquad \gamma^{\pi_n}_{*,n}(t) := \frac{1}{2}\sum_{i=1}^n \pi^i_n(t)(1-\pi^i_n(t))(\sigma^i_n(t))^2,
\end{equation*}
are respectively called the {\em growth rate} and the {\em excess growth rate} of the portfolio. 

Clearly, the growth rate of the portfolio writes as the average of the growth rates of the stocks contained in the portfolio, with weights given by the portfolio, plus the excess growth rate. The latter rewrites as the average of the variances of the stocks contained in the portfolio, minus the variance of the wealth process. Thus, the variance reduction due to diversification in the portfolio is exactly measured by the excess growth rate $\gamma^{\pi_n}_{*,n}(t) \geq 0$.

\subsubsection{Diversity weighted portfolios} For all $p \geq 0$, we now define the {\em $p$-diversity weighted} portfolio $\pi^p_n = (\pi^{p,1}_n(t), \ldots, \pi^{p,n}_n(t))_{t \geq 0}$ by
\begin{equation*}
  \forall t \geq 0, \quad \forall i \in \{1, \ldots, n\}, \qquad \pi^{p,i}_n(t) := \frac{(X^i_n(t))^p}{(X^1_n(t))^p + \cdots + (X^n_n(t))^p}.
\end{equation*}
The associated wealth process is denoted by $(Z^p_n(t))_{t \geq 0}$ and the growth rate and excess growth rate processes of the portfolio are respectively denoted by $(\gamma^p_n(t))_{t \geq 0}$, $(\gamma^p_{*,n}(t))_{t \geq 0}$. The parameter $p$ is called the {\em diversity index}. 

Certainly, the choice $p=0$ corresponds to the equally weighted portfolio, while the choice $p=1$ is the market portfolio. For $0<p<1$, the $p$-diversity weighted portfolio interpolates between the equally weighted portfolio and the market portfolio, and it is {\em functionally generated} by a {\em measure of diversity} in the sense of Fernholz~\cite[Section~3.4]{fernholz}. Let us also mention that diversity weighted portfolios, with $p=0.76$, were used in actual portfolio managing strategies for the S\&P 500 Index~\cite[Section~7.2]{fernholz}.

\subsubsection{Long-term growth rate and performance} Following~\cite[Section~1.3]{fernholz}, the growth rate of a portfolio measures its long-term performance, in the sense that
\begin{equation*}
  \lim_{T \to +\infty} \frac{1}{T} \left(\log Z_n^{\pi_n}(T) - \int_{t=0}^T \gamma^{\pi_n}_n(t)\dd t\right) = 0, \quad \text{almost surely}.
\end{equation*}
As a consequence, the study of the long time behaviour of the processes $(\gamma^{\pi_n}_n(t))_{t \geq 0}$ and $(\gamma^{\pi_n}_{*,n}(t))_{t \geq 0}$ arises as a natural question with respect to practical situations. As far as the asymptotic behaviour of portfolio rules in large markets is concerned, previous studies such as~\cite{banner} measured the performance of a sequence of portfolio rules $\{(\pi_n(t))_{t \geq 0}, n \geq 1\}$ by analyzing the asymptotic long-term growth rate 
\begin{equation*}
  \Gamma := \lim_{n \to +\infty} \lim_{t \to +\infty} \gamma^{\pi_n}_n(t).
\end{equation*}
The latter was computed for the Atlas model~\eqref{eq:atlas}, with constant variance coefficients $(\sigma^1_n)^2 = \cdots = (\sigma^n_n)^2 > 0$, or linearly decreasing coefficients $(\sigma^j_n)^2 = \ua + s^2 (n-j)$, $\ua > 0$, $s^2 > 0$, which matches the empirical observation of~\cite[Figure~13.6]{ferkar}. For the equally weighted portfolio and the market portfolio, exact expressions were derived. For $p$-diversity weighted portfolios with $p \in (0,1)$, the long-term growth rate and excess growth rate were explicited in terms of the stationary distribution of the market portfolio. At that time, very little was known about this stationary distribution, so that the authors had to resort to the so-called {\em certainty-equivalent approximation} to describe the large market limit of the long-term growth rate and excess growth rate. Still, in all cases, it was observed that, for large markets, the equally weighted portfolio outperforms diversity weighted portfolios, and in particular, beats the market.

\subsubsection{Growth rates and capital measure} We provide a detailed study of the long-term asymptotic growth rate 
\begin{equation*}
  G^p := \lim_{t \to +\infty} \lim_{n \to +\infty} \gamma^p_n(t)
\end{equation*} 
for $p$-diversity weighted portfolios, which brings forth results in a very synthetic fashion. On account of Subsection~\ref{ss:interv}, we note that, at least in the case of a constant variance function, this approach is equivalent to~\cite{banner}.

Let us remark that, by the definition of $p$-diversity weighted portfolios and due to~\eqref{eq:mfgammasigma} and~\eqref{eq:Pin}, for all $p \geq 0$, the growth rate and excess growth rate of $p$-diversity weighted portfolios write
\begin{equation}\label{eq:growthrates}
  \begin{aligned}
    & \gamma^p_n(t) = \crochet{\gamma}{\Pi^p_n(t)} + \gamma^p_{*,n}(t),\\
    & \gamma^p_{*,n}(t) = \frac{1}{2}\crochet{\sigma^2}{\Pi^p_n(t)} - \frac{1}{2} \sum_{j=1}^n \left(\frac{\ex{p Y^{(j)}_n(t)}}{\ex{p Y^1_n(t)} + \cdots + \ex{p Y^n_n(t)}}\right)^2\sigma^2\left(\frac{j}{n}\right),
  \end{aligned}
\end{equation}
while the quadratic variation of $(\log Z^p_n(t))_{t \geq 0}$ writes
\begin{equation}\label{eq:quadvar}
  \langle \log Z^p_n \rangle (t) = \sum_{j=1}^n \sigma^2\left(\frac{j}{n}\right) \int_{s=0}^t \left(\frac{\ex{p Y^{(j)}_n(s)}}{\ex{p Y^1_n(s)} + \cdots + \ex{p Y^n_n(s)}}\right)^2\dd s.
\end{equation}
We shall prove below that the quantity
\begin{equation*}
  \sum_{j=1}^n \left(\frac{\ex{p Y^{(j)}_n(t)}}{\ex{p Y^1_n(t)} + \cdots + \ex{p Y^n_n(t)}}\right)^2\sigma^2\left(\frac{j}{n}\right)
\end{equation*}
is of order $1/n$, so that the analysis of the long-term asymptotic portfolio performance only relies on the analysis of the long-term asymptotic weighted capital measure.

\subsection{Long-term asymptotic growth rates}\label{ss:G} We first derive laws of large numbers for the processes $(\gamma^p_n(t))_{t \geq 0}$, $(\gamma^p_{*,n}(t))_{t \geq 0}$ and $(Z^p_n(t))_{t \geq 0}$, based on Proposition~\ref{prop:lln:Pi}. The proof of Lemma~\ref{lem:lln:gamma} is postponed to Appendix~\ref{app:lln}.

\begin{lem}\label{lem:lln:gamma}
  Let us assume that the conditions of Proposition~\ref{prop:lln:Pi} are satisfied. Let us fix $T>0$ and $q \in [1,+\infty)$. Then, for all $p \geq 0$, when $n$ grows to infinity,
  \begin{itemize}
    \item the growth rate $(\gamma^p_n(t))_{t \in [0,T]}$ and excess growth rate $(\gamma^p_{*,n}(t))_{t \in [0,T]}$ respectively converge, in $L^q(C([0,T],\R))$, to the deterministic processes $(\gamma^p(t))_{t \in [0,T]}$ and $(\gamma^p_*(t))_{t \in [0,T]}$ defined by
    \begin{equation}\label{eq:gammapgammapstar}
      \forall t \geq 0, \qquad \gamma^p(t) = \crochet{b}{\Pi^p(t)}, \qquad \gamma^p_*(t) = \frac{1}{2}\crochet{\sigma^2}{\Pi^p(t)},
    \end{equation}
    where $b$ is the {\em rate of return function} defined by
    \begin{equation*}
      b(u) := \gamma(u) + \frac{1}{2} \sigma^2(u);
    \end{equation*}
  
    \item the wealth process $(Z^p_n(t))_{t \in [0,T]}$ converges, in $L^q(C([0,T],\R))$, to the deterministic process $(Z^p(t))_{t \in [0,T]}$ defined by
    \begin{equation}\label{eq:Zpinf}
      \forall t \geq 0, \qquad \log Z^p(t) = \int_{s=0}^t \gamma^p(s)\dd s.
    \end{equation}
  \end{itemize}
\end{lem}

The deterministic processes $(\gamma^p(t))_{t \geq 0}$, $(\gamma^p_*(t))_{t \geq 0}$ and $(Z^p(t))_{t \geq 0}$ shall be respectively called the asymptotic growth rate, the asymptotic excess growth rate and the asymptotic wealth process associated with $p$-diversity weighted portfolios. Their long time behaviour is determined by Proposition~\ref{prop:phtrans} and Lemma~\ref{lem:lln:gamma} as follows.

\begin{lem}\label{lem:longtime:gamma}
  Let us assume that the conditions of Proposition~\ref{prop:phtrans} hold, and recall the definition~\eqref{eq:pc} of the critical diversity index $\pc \geq 0$.
  \begin{itemize}
    \item For all $p \in [0, \pc)$,
    \begin{equation*}
      G^p := \lim_{t \to +\infty} \gamma^p(t) = \crochet{b}{\bPi^p}, \qquad G_*^p := \lim_{t \to +\infty} \gamma_*^p(t) = \frac{1}{2}\crochet{\sigma^2}{\bPi^p}.
    \end{equation*}
    \item For all $p > \pc$,
    \begin{equation*}
      G^p := \lim_{t \to +\infty} \gamma^p(t) = b(1), \qquad G_*^p := \lim_{t \to +\infty} \gamma_*^p(t) = \frac{1}{2}\sigma^2(1).
    \end{equation*}
  \end{itemize}
\end{lem}

The quantities $G^p$ and $G^p_*$ shall be respectively called the long-term asymptotic growth rate and the long-term asymptotic excess growth rates. When $p=\pc$, the limits of $\gamma^{\pc}(t)$ and $\gamma^{\pc}_*(t)$ when $t \to +\infty$ may not exist, therefore we define
\begin{equation*}
  G^{\pc} := \limsup_{t \to +\infty} \gamma^{\pc}(t), \qquad G_*^{\pc} := \limsup_{t \to +\infty} \gamma_*^{\pc}(t).
\end{equation*}
Following Subsection~\ref{ss:spt}, the performance of the $p$-diversity weighted portfolio is measured by its long-term asymptotic growth rate, therefore we shall look for values of the diversity index $p$ for which $G^p$ is maximal. We will most frequently obtain values for $p$ that are not equal to $\pc$, so that defining $G^{\pc}$ as the superior limit of $\gamma^{\pc}(t)$ ensures the actual optimality of $p$.

Let us finally mention that Proposition~\ref{prop:phtrans} ensures that the functions $p \mapsto G^p$ and $p \mapsto G^p_*$ are continuous on $[0,\pc)$, constant on $(\pc,+\infty)$, and satisfy
\begin{equation*}
  \liminf_{p \to \pc} G^p \leq G^{\pc} \leq \limsup_{p \to \pc} G^p, \qquad \liminf_{p \to  \pc} G_*^p \leq G_*^{\pc} \leq \limsup_{p \to \pc} G_*^p.
\end{equation*}

\subsection{Montonicity criterion}\label{ss:monoton} We first address the monotonicity of the functions $p \mapsto G^p$ and $p \mapsto G^p_*$, based on the following lemma.

\begin{lem}\label{lem:monoton}
  Let us assume that the conditions of Proposition~\ref{prop:lln:Pi} hold and fix a continuous function $f : [0,1] \to \R$. If $f$ is monotonic on $[0,1]$, then, for all $t \geq 0$, the function $p \mapsto \crochet{f}{\Pi^p(t)}$ has the same monotonicity on $[0,+\infty)$.
\end{lem}
\begin{proof}
  Let us fix a continuous function $f : [0,1] \to \R$. By~\eqref{eq:EY} and the Leibniz integral rule, for all $t \geq 0$, the function 
  \begin{equation*}
    p \mapsto \int_{u=0}^1 \ex{p F_t^{-1}(u)}f(u)\dd u
  \end{equation*}
  is $C^1$ on $[0,+\infty)$, and its derivative writes
  \begin{equation*}
    \frac{\dd}{\dd p} \int_{u=0}^1 \ex{p F_t^{-1}(u)}f(u)\dd u = \int_{u=0}^1 F_t^{-1}(u) \ex{p F_t^{-1}(u)}f(u)\dd u,
  \end{equation*}
  from which it easily follows that the function $p \mapsto \scrochet{f}{\Pi^p(t)}$ is $C^1$ on $[0,+\infty)$ and
  \begin{equation*}
    \frac{\dd}{\dd p} \crochet{f}{\Pi^p(t)} = \crochet{F_t^{-1}f}{\Pi^p(t)} - \crochet{F_t^{-1}}{\Pi^p(t)}\crochet{f}{\Pi^p(t)} = \Cov(F_t^{-1}(U), f(U)),
  \end{equation*}
  where the random variable $U \in [0,1]$ is distributed according to $\Pi^p(t)$. Since $F_t^{-1}$ is nondecreasing, the result follows from a classical comonotonicity argument.
\end{proof}

We can now derive the following monotonicity criterion for the long-term asymptotic growth rate and excess growth rate.

\begin{cor}\label{cor:monoton}
  Let us assume that the conditions of Proposition~\ref{prop:phtrans} hold.
  \begin{itemize}
    \item If the rate of return function $b$ is monotonic on $[0,1]$, then the function $p \mapsto G^p$ has the same monotonicity on $[0,+\infty)$.
    \item If the variance function $\sigma^2$ is monotonic on $[0,1]$, then the function $p \mapsto G^p_*$ has the same monotonicity on $[0,+\infty)$.
  \end{itemize}
\end{cor}

\subsection{The reduction formula}\label{ss:reduction} We complete the monotonicity criterion of Corollary~\ref{cor:monoton} by the following reduction formula expressing the long-term asymptotic growth rate in terms of the long-term asymptotic excess growth rate in the subcritical phase.

\begin{prop}\label{prop:IPP}
  Let us assume that the conditions of Proposition~\ref{prop:phtrans} hold, and that $\pc>0$. Then, for all $p \in [0, \pc)$,
  \begin{equation*}
    G^p = (1-p) G^p_* + g.
  \end{equation*}
\end{prop}
\begin{proof}
  Let us assume that $\pc>0$ and fix $p \in [0,\pc)$. Using Lemma~\ref{lem:longtime:gamma}, we first write
  \begin{equation*}
    G^p = \crochet{\gamma}{\bPi^p} + \frac{1}{2}\crochet{\sigma^2}{\bPi^p} = \crochet{\tilde{\gamma}}{\bPi^p} + g + G^p_*,
  \end{equation*}
  where we recall that $\tilde{\gamma}(u) = \gamma(u) - g$. Thanks to the first-order analysis of $\Psi$ carried out in Remark~\ref{rk:tail},
  \begin{equation*}
    \lim_{u \dto 0} \ex{p\Psi(u)} (\Gamma(u)-gu) = 0, \quad \lim_{u \uto 1} \ex{p\Psi(u)} (\Gamma(u)-gu) = \lim_{u \uto 1} (g-\gamma(1))(1-u)^{1-p/\pc} = 0,
  \end{equation*}
  so that integrating by parts yields
  \begin{equation*}
    \int_{u=0}^1 \ex{p \Psi(u)} \tilde{\gamma}(u) \dd u = - \int_{u=0}^1 p\Psi'(u) \ex{p\Psi(u)} (\Gamma(u)-gu)\dd u = -\frac{p}{2} \int_{u=0}^1 \ex{p\Psi(u)}\sigma^2(u)\dd u,
  \end{equation*}
  hence $\scrochet{\tilde{\gamma}}{\bPi^p} = - p G^p_*$.
\end{proof}

\begin{rk}\label{rk:reduction}
  In the supercritical phase, elementary algebra allows to derive a similar reduction formula, where $p$ has to be replaced with $\pc$, namely $G^p = (1-\pc) G^p_* + g$, for all $p>\pc$. Both formulas rewrite in a compact form as
  \begin{equation*}
    \forall p\not=\pc, \qquad G^p = (1-p\wedge\pc) G^p_* + g,
  \end{equation*}
  and this also holds true for $p=\pc$ as soon as at least one of the functions $p \mapsto G^p$ or $p \mapsto G^p_*$ is continuous at $\pc$.
\end{rk}

\subsection{Performance of the equally weighted and the market portfolio}\label{ss:perf} Let us apply the results of Proposition~\ref{prop:IPP} to describe the performance of the equally weighted and the market portfolio.

{\em Equally weighted portfolio:} the long-term asymptotic growth rate writes
\begin{equation*}
  G^{(0)} = G^{(0)}_* + g = \frac{1}{2}\int_{u=0}^1 \sigma^2(u)\dd u + g > g,
\end{equation*}
so that the equally weighted portfolio grows faster than the market mean growth rate $g$, by a factor depending only on the volatility structure of the market.
  
{\em Market portfolio:} if $\pc > 1$, then the long-term asymptotic growth rate writes $G^{(1)} = g$, so that the market portfolio grows at the market mean growth rate. If $\pc < 1$, then 
\begin{equation*}
  G^{(1)} = (1-\pc) \frac{\sigma^2(1)}{2} + g > g,
\end{equation*} 
so that the market portfolio grows faster than the market mean growth rate, by a factor depending on both the growth rate function and the variance function of the market.

\subsection{Optimal selection of portfolios and volatility structure}\label{ss:optimal} We now combine the results of Corollary~\ref{cor:monoton} and Proposition~\ref{prop:IPP} to select the portfolio rule with best performance, depending on the volatility structure of the market. We sum up our results in Conclusions~\eqref{conc:1}, \eqref{conc:2} and~\eqref{conc:3}. 

Let us first assume that the variance function $\sigma^2$ is nonincreasing, which matches the volatility size effect. Then, Corollary~\ref{cor:monoton} implies that the long-term asymptotic excess growth rate $G^p_*$ is nonincreasing on $[0,+\infty)$. Using the reduction formula of Proposition~\ref{prop:IPP}, we deduce that the long-term asymptotic growth rate $G^p$ is nonincreasing on $[0,+\infty)$, therefore it is maximal for $p=0$.
\begin{enumerate}[label=(C1), ref=C1]
  \item\label{conc:1} If the variance function is nonincreasing, then the equally weighted portfolio is optimal among $p$-diversity weighted portfolios.
\end{enumerate}

A particular case of a nonincreasing variance function is the case of a constant variance function. Then, by Remark~\ref{rk:reduction}, for all $p \not= \pc$,
\begin{equation}\label{eq:varconst}
  G^p = (1-p \wedge \pc) \frac{\sigma^2}{2} + g.
\end{equation}
The expression above has the same right and left limits in $\pc$, so that, by Remark~\ref{rk:reduction}, the formula~\eqref{eq:varconst} is actually valid for all $p \in [0,+\infty)$.
\begin{enumerate}[label=(C2), ref=C2]
  \item\label{conc:2} If the variance function is constant, then, for all $p \geq 0$, the long-term asymptotic growth rate of the $p$-diversity weighted portfolio is given by the formula~\eqref{eq:varconst}.
\end{enumerate}

We finally look for conditions on the market model to produce a situation in which the equally weighted portfolio is {\em not} optimal among $p$-diversity weighted portfolios. On account of Corollary~\ref{cor:monoton}, this is the case if the rate of return function $b$ is increasing on $[0,1]$. In such a situation, and under Assumptions~\eqref{hyp:UE} and~\eqref{hyp:E1},
\begin{equation*}
  b(1) = \gamma(1) + \frac{1}{2}\sigma^2(1) > \gamma(0) + \frac{1}{2}\sigma^2(0) > g,
\end{equation*}
so that $\pc < 1$. Then, using the results of Subsection~\ref{ss:perf},
\begin{equation*}
  G^{(0)} = \int_{u=0}^1 b(u)\dd u < b(1) = G^{(1)},
\end{equation*}
that is to say, the market portfolio outperforms the equally weighted portfolio --- and it is actually optimal among all $p$-diversity weighted portfolios.

\begin{expl}\label{expl:portfolio}
  Let us specify an example of a model where the market portfolio is optimal. We use the growth rate function introduced in the mean-field approximation of the Atlas model of Remark~\ref{rk:atlas}, $\gamma(u) = \gamma_{\alpha}(u) = g(\alpha+1)(1-u)^{\alpha}$, with $\alpha > 0$ to be specified below. Recall that this growth rate function satisfies Assumptions~\eqref{hyp:E1} and~\eqref{hyp:E2}. We now choose the variance function $\sigma^2$ in order to satisfy the uniform ellipticity assumption~\eqref{hyp:UE} and to ensure that the rate of return function $b = \gamma + \sigma^2/2$ is increasing; for instance, we let
  \begin{equation*}
    \sigma^2(u) = 2\left(C + u - \gamma_{\alpha}(u)\right),
  \end{equation*}
  with $C=1+g(\alpha+1)$, see Figure~\ref{fig:sigma2}. Then, for all $\alpha>0$, $b(u) = C+u$ is increasing and $\sigma^2$ satisfies the uniform ellipticity assumption~\eqref{hyp:UE}. We now take $\alpha>2$ to ensure that the regularity assumptions on $\gamma$ and $\sigma^2$ required in Theorem~\ref{theo:equi} are fulfilled. This completes the construction of our model, and effectively provides an instance of a mean-field Atlas market model where the market portfolio outperforms the equally weighted portfolio.
  \begin{figure}[ht]
    \includegraphics[width=15cm]{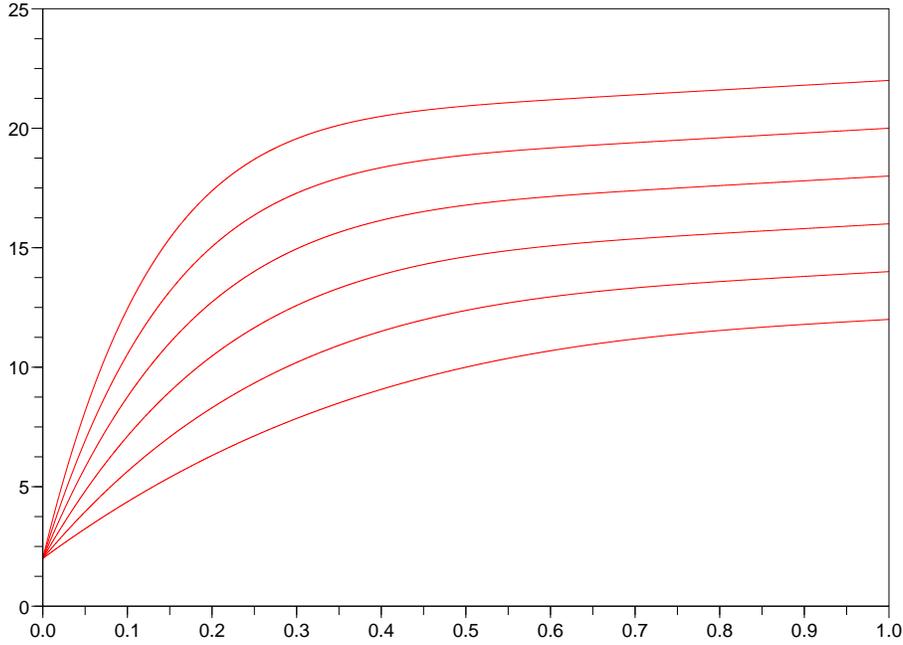}
    \caption{The variance function $\sigma^2$ for $g=1$ and $\alpha$ taking integer values ranging from $3$ (lowest curve) to $8$ (highest curve).}
    \label{fig:sigma2}
  \end{figure}
\end{expl}

Example~\ref{expl:portfolio} leads to the following conclusion.
\begin{enumerate}[label=(C3), ref=C3]
  \item\label{conc:3} One can exhibit an example of a model where the market portfolio is optimal among all $p$-diversity weighted portfolios. It is necessary that, in such a model, small stocks have a smaller variance than large stocks, so that the volatility size effect is violated.
\end{enumerate}


\appendix

\section{Proofs of Proposition~\ref{prop:lln:Pi} and Lemma~\ref{lem:lln:gamma}}\label{app:lln}

This appendix is dedicated to the proofs of Proposition~\ref{prop:lln:Pi} and of the related Lemma~\ref{lem:lln:gamma}. We first prove~\eqref{eq:supexpY}.

\begin{proof}[Proof of~\eqref{eq:supexpY}]
  Let us fix $T>0$ and $p \geq 0$. By Theorem~\ref{theo:chaos}, for all $t \in [0,T]$,
  \begin{equation*}
    \begin{aligned}
      \Z^p(t) := \int_{u=0}^1 \ex{pF_t^{-1}(u)}\dd u = \Exp\left(\ex{pY(t)}\right) & = \Exp\left(\ex{pY(0) + p\int_{s=0}^t \gamma(F_s(Y(s)))\dd s + p\int_{s=0}^t \sigma(F_s(Y(s)))\dd B(s)}\right)\\
      & \leq \ex{p ||\gamma||_{\infty} T} \Exp\left(\ex{p Y(0)} \Exp\left(\ex{p\int_{s=0}^t \sigma(F_s(Y(s)))\dd B(s)} \Big| Y(0)\right)\right)\\
      & \leq \ex{p ||\gamma||_{\infty} T + (p^2/2)||\sigma^2||_{\infty}T} \int_{y \in \R} \ex{py}m(\dd y) =: C^p_T,
    \end{aligned}
  \end{equation*}
  and Assumption~\eqref{hyp:H} ensures that the right-hand side above is finite.
\end{proof}

We now address the second part of Proposition~\ref{prop:lln:Pi}. In this purpose, we first state the following auxiliary lemma.

\begin{lem}\label{lem:lln}
  Under the assumptions of Proposition~\ref{prop:lln:Pi}, for all $T>0$ and $p \geq 0$, for all continuous function $f : [0,1] \to \R$, 
  \begin{equation*}
    \lim_{n \to +\infty} \Exp\left(\sup_{t \in [0,T]} \left|\hat{f}^p_n(t) - \hat{f}^p(t)\right|\right) = 0,
  \end{equation*}
  where, for all $t \geq 0$,
  \begin{equation*}
    \hat{f}^p_n(t) := \frac{1}{n} \sum_{j=1}^n \ex{pY^{(j)}_n(t)} f\left(\frac{j}{n}\right), \quad \hat{f}^p(t) := \int_{u=0}^1 \ex{p F_t^{-1}(u)} f(u)\dd u.
  \end{equation*}
\end{lem}

Before giving the proof of Lemma~\ref{lem:lln}, let us explain how to complete the proof of Proposition~\ref{prop:lln:Pi}: let us fix a continuous function $f : [0,1] \to \R$, $p \geq 0$ and $T>0$. Then, for all $t \in [0,T]$,
\begin{equation*}
  \crochet{f}{\Pi^p_n(t)} = \frac{\hat{f}_n^p(t)}{\hat{\mathbf{1}}_n^p(t)},
\end{equation*}
where we denote by $\mathbf{1}$ the constant function equal to $1$. Combining Lemma~\ref{lem:lln} with the Slutsky theorem, and using the continuity of the mapping 
\begin{equation*}
  \left((x(t))_{t \in [0,T]}, (y(t))_{t \in [0,T]}\right) \mapsto \left(\frac{x(t)}{y(t)}\right)_{t \in [0,T]}
\end{equation*}
at all point $((x(t))_{t \in [0,T]}, (y(t))_{t \in [0,T]}) \in (C([0,T],\R))^2$ such that, for all $t \in [0,T]$, $y(t) \not= 0$, we deduce that the sequence of processes $(\scrochet{f}{\Pi^p_n(t)})_{t \in [0,T]}$ converges in probability, in $C([0,T],\R)$, to the process $(\scrochet{f}{\Pi^p(t)})_{t \in [0,T]}$ introduced in Proposition~\ref{prop:lln:Pi}. Thanks to the elementary bound
\begin{equation*}
  \forall t \geq 0, \qquad \left|\crochet{f}{\Pi^p_n(t)}\right| \leq ||f||_{\infty},
\end{equation*}
we conclude that the convergences above also hold in $L^q(C([0,T],\R))$, for all $q \in [1, +\infty)$.

\begin{proof}[Proof of Lemma~\ref{lem:lln}]
  Let us fix $T>0$ and $p \geq 0$. The key observation is that, for all $t \in [0,T]$, the reordered vector $(Y^{(1)}_n(t), \ldots, Y^{(n)}_n(t))$ writes
  \begin{equation*}
    \forall j \in \{1, \ldots, n\}, \quad \forall u \in [(j-1)/n,j/n), \qquad Y^{(j)}_n(t) = (H*\nu^n_t)^{-1}(u), 
  \end{equation*}
  where $(H*\nu^n_t)^{-1}$ refers to the pseudo-inverse of the empirical cumulative distribution function of $Y^1_n(t), \ldots, Y^n_n(t)$. Therefore, for all continuous function $f : [0,1] \to \R$, for all $t \in [0,T]$,
  \begin{equation*}
    \begin{aligned}
      & \left|\hat{f}^p_n(t) - \hat{f}^p(t) \right| = \left|\sum_{j=1}^n \int_{u=(j-1)/n}^{j/n} \left(\ex{p(H*\nu^n_t)^{-1}(u)} f \left(\frac{j}{n}\right) - \ex{pF_t^{-1}(u)}f(u)\right)\dd u\right|\\
      & \qquad \leq ||f||_{\infty} \int_{u=0}^1 \left|\ex{p(H*\nu^n_t)^{-1}(u)} - \ex{pF_t^{-1}(u)}\right|\dd u + \sum_{j=1}^n \int_{u=(j-1)/n}^{j/n} \ex{pF_t^{-1}(u)} \left|f \left(\frac{j}{n}\right) - f(u)\right|\dd u.
    \end{aligned}
  \end{equation*}
  Combining the uniform continuity of $f$ with~\eqref{eq:supexpY} yields
  \begin{equation}\label{eq:step0}
    \lim_{n \to +\infty} \sup_{t \in [0,T]} \sum_{j=1}^n \int_{u=(j-1)/n}^{j/n} \ex{pF_t^{-1}(u)} \left|f \left(\frac{j}{n}\right) - f(u)\right|\dd u = 0.
  \end{equation}
  
  We now fix $M>0$ and write
  \begin{equation*}
    \int_{u=0}^1 \left|\ex{p(H*\nu^n_t)^{-1}(u)} - \ex{pF_t^{-1}(u)}\right|\dd u = I_n^M(t) + J_n^M(t),
  \end{equation*}
  where
  \begin{equation*}
    \begin{aligned}
      & I_n^M(t) := \int_{u=0}^1 \left(\left|\ex{p(H*\nu^n_t)^{-1}(u)} - \ex{pF_t^{-1}(u)}\right| \wedge M\right)\dd u,\\
      & J_n^M(t) := \int_{u=0}^1 \left[\left|\ex{p(H*\nu^n_t)^{-1}(u)} - \ex{pF_t^{-1}(u)}\right|-M\right]^+\dd u,
    \end{aligned}
  \end{equation*}
  with $[x]^+ := x \vee 0$. In Step~1 below, we shall establish that
  \begin{equation}\label{eq:step1}
    \forall M > 0, \qquad \lim_{n \to +\infty} \Exp\left(\sup_{t \in [0,T]} I^M_n(t)\right) = 0,
  \end{equation}
  while Step~2 is dedicated to the proof of
  \begin{equation}\label{eq:step2}
    \lim_{M \to +\infty} \sup_{n \geq 1} \Exp\left(\sup_{t \in [0,T]} J_n^M(t)\right) = 0.
  \end{equation}
  Then, it follows from~\eqref{eq:step0} and~\eqref{eq:step1} that, for all $M > 0$,
  \begin{equation*}
    \limsup_{n \to +\infty} \Exp\left(\sup_{t \in [0,T]} \left|\hat{f}^p_n(t) - \hat{f}^p(t) \right|\right) \leq \sup_{n \geq 1} \Exp\left(\sup_{t \in [0,T]} J_n^M(t)\right),
  \end{equation*}
  and the conclusion stems from~\eqref{eq:step2}.
  
  {\em Step~1.} Let us first note that, by Remark~\ref{rk:chaos}, for all $t \in [0,T]$,
  \begin{equation*}
    \lim_{n \to +\infty} \Exp\left(I^M_n(t)\right) = 0,
  \end{equation*}
  so that, by the Slutsky theorem, the continuous process $(I^M_n(t))_{t \in [0,T]}$ converges, in finite dimensional distribution, to $0$. Taking for granted that the sequence of the laws of the processes $(I^M_n(t))_{t \in [0,T]}$, $n \geq 1$ is tight, the Prohorov theorem~\cite[Theorem~8.1, p.~54]{billingsley} implies that this sequence converges to $0$ in probability in $C([0,T],\R)$. Then,~\eqref{eq:step1} follows from the uniform boundedness of the process $(I^M_n(t))_{t \in [0,T]}$ with respect to $n$.
  
  To complete this step, it remains to prove that the sequence of the laws of the processes $(I^M_n(t))_{t \in [0,T]}$, $n \geq 1$ is tight. First, the sequence $\{I^M_n(0), n \geq 1\}$ is bounded by $M$, and therefore the sequence of the laws of $I^M_n(0)$, $n \geq 1$ is tight. Thus, by the Kolmogorov criterion~\cite[Theorem~12.3, p.~95]{billingsley}, it now suffices to exhibit $\alpha \geq 0$, $\delta > 0$ and $C \geq 0$ such that, for all $t,s \in [0,T]$,
  \begin{equation*}
    \forall n \geq 1, \qquad \Exp\left(|I^M_n(t)-I^M_n(s)|^{\alpha}\right) \leq C |t-s|^{1+\delta}.
  \end{equation*}
  We first use the chain of elementary inequalities
  \begin{equation*}
    \begin{aligned}
      \forall x_1,x_2,y_1,y_2 \in \R, \qquad \big||x_1-y_1| \wedge M - |x_2-y_2| \wedge M\big| & \leq \big||x_1-y_1|-|x_2-y_2|\big|\\
      & \leq |x_1-x_2| + |y_1-y_2|
    \end{aligned}
  \end{equation*}
  to rewrite, for all $t,s \in [0,T]$ such that $s \leq t$,
  \begin{equation*}
    |I^M_n(t)-I^M_n(s)| \leq \int_{u=0}^1 \left|\ex{p(H*\nu_t^n)^{-1}(u)}-\ex{p(H*\nu_s^n)^{-1}(u)}\right|\dd u + \int_{u=0}^1 \left|\ex{pF_t^{-1}(u)}-\ex{pF_s^{-1}(u)}\right|\dd u.
  \end{equation*}
  Let us now fix $\alpha > 2$. By the Jensen inequality, the inequality above yields
  \begin{equation}\label{eq:step1:jensen}
    \begin{aligned}
      & |I^M_n(t)-I^M_n(s)|^{\alpha}\\
      & \qquad \leq 2^{\alpha-1}\left(\int_{u=0}^1 \left|\ex{p(H*\nu_t^n)^{-1}(u)}-\ex{p(H*\nu_s^n)^{-1}(u)}\right|^{\alpha}\dd u + \int_{u=0}^1 \left|\ex{pF_t^{-1}(u)}-\ex{pF_s^{-1}(u)}\right|^{\alpha}\dd u\right).
    \end{aligned}
  \end{equation}
  Let us address the first term in the right-hand side of~\eqref{eq:step1:jensen}. Using the Jensen inequality again,
  \begin{equation*}
    \begin{aligned}
      & \int_{u=0}^1 \left|\ex{p(H*\nu_t^n)^{-1}(u)}-\ex{p(H*\nu_s^n)^{-1}(u)}\right|^{\alpha}\dd u = \frac{1}{n} \sum_{i=1}^n \left|\ex{pY^i_n(t)}-\ex{pY^i_n(s)}\right|^{\alpha}\\
      & \leq 2^{\alpha-1} \left(\frac{1}{n} \sum_{i=1}^n \left|\int_{r=s}^t \ex{pY^i_n(r)}\left(p\gamma^i_n(r) + \frac{p^2}{2}(\sigma^i_n(r))^2\right)\dd r\right|^{\alpha} + \frac{1}{n} \sum_{i=1}^n \left|\int_{r=s}^t p\ex{pY^i_n(r)}\sigma_n^i(r)\dd B^i(r)\right|^{\alpha}\right).
    \end{aligned}
  \end{equation*}
  
  On the one hand,
  \begin{equation*}
    \begin{aligned}
      & \Exp\left(\left|\int_{r=s}^t \ex{pY^i_n(r)}\left(p\gamma^i_n(r) + \frac{p^2}{2}(\sigma^i_n(r))^2\right)\dd r\right|^{\alpha}\right)\\
      & \qquad \leq \left(p||\gamma||_{\infty} + \frac{p^2}{2}||\sigma^2||_{\infty}\right)^{\alpha}(t-s)^{\alpha-1}\int_{r=s}^t \Exp\left(\ex{\alpha p Y^i_n(r)}\right)\dd r,
    \end{aligned}
  \end{equation*}
  and by the same arguments as in the proof of~\eqref{eq:supexpY},
  \begin{equation*}
    \forall r \in [0,T], \qquad \Exp\left(\ex{\alpha p Y^i_n(r)}\right) \leq C^{\alpha p}_T,
  \end{equation*}
  where the constant $C^{\alpha p}_T$ does not depend on $n$. As a consequence,
  \begin{equation*}
    \Exp\left(\frac{1}{n} \sum_{i=1}^n \left|\int_{r=s}^t \ex{pY^i_n(r)}\left(p\gamma^i_n(r) + \frac{p^2}{2}(\sigma^i_n(r))^2\right)\dd r\right|^{\alpha}\right) \leq C^{\alpha p}_T\left(p||\gamma||_{\infty} + \frac{p^2}{2}||\sigma^2||_{\infty}\right)^{\alpha}(t-s)^{\alpha}.
  \end{equation*}
  
  On the other hand, the Burkholder-Davis-Gundy inequality implies that there exists $K \geq 0$ depending only on $\alpha$ such that
  \begin{equation*}
    \begin{aligned}
      \Exp\left(\left|\int_{r=s}^t p\ex{pY^i_n(r)}\sigma_n^i(r)\dd B^i(r)\right|^{\alpha}\right) & \leq K \Exp\left(\left|\int_{r=s}^t p^2\ex{2pY^i_n(r)}(\sigma_n^i(r))^2\dd r\right|^{\alpha/2}\right)\\
      & \leq Kp^2||\sigma^2||_{\infty} (t-s)^{\alpha/2-1} \int_{r=s}^t \Exp\left(\ex{\alpha p Y^i_n(r)}\right)\dd r\\
      & \leq Kp^2||\sigma^2||_{\infty} (t-s)^{\alpha/2} C^{\alpha p}_T.
    \end{aligned}
  \end{equation*}
  
  As a conclusion, there exists $C' \geq 0$ such that
  \begin{equation*}
    \Exp\left(\int_{u=0}^1 \left|\ex{p(H*\nu_t^n)^{-1}(u)}-\ex{p(H*\nu_s^n)^{-1}(u)}\right|^{\alpha}\dd u\right) \leq C'|t-s|^{\alpha/2}.
  \end{equation*}
  
  The second term of~\eqref{eq:step1:jensen} rewrites
  \begin{equation*}
    \int_{u=0}^1 \left|\ex{pF_t^{-1}(u)}-\ex{pF_s^{-1}(u)}\right|^{\alpha}\dd u = \Exp\left(\left|\ex{pF_t^{-1}(U)}-\ex{pF_s^{-1}(U)}\right|^{\alpha}\right),
  \end{equation*} 
  where $U$ is a uniform random variable on $[0,1]$. Note that $\ex{pF_t^{-1}(U)}$ has the same marginal distribution as $\ex{pY(t)}$, and $\ex{pF_s^{-1}(U)}$ has the same marginal distribution as $\ex{pY(s)}$. By~\eqref{eq:coupl} and~\eqref{eq:wasserstein},
  \begin{equation*}
    \int_{u=0}^1 \left|\ex{pF_t^{-1}(u)}-\ex{pF_s^{-1}(u)}\right|^{\alpha}\dd u \leq \Exp\left(\left|\ex{pY(t)}-\ex{pY(s)}\right|^{\alpha}\right),
  \end{equation*} 
  and the same arguments as for the first term in the right-hand side of~\eqref{eq:step1:jensen} allow us to conclude that the right-hand side above is bounded by $C'(t-s)^{\alpha/2}$. As a conclusion,
  \begin{equation*}
    \Exp\left(|I^M_n(t)-I^M_n(s)|^{\alpha}\right) \leq 2^{\alpha-1}C'(t-s)^{\alpha/2},
  \end{equation*}
  therefore the sequence of the laws of $(I^M_n(t))_{t \geq 0}$, $n \geq 1$ is tight.
  
  {\em Step 2.} Using the chain of elementary inequalities
  \begin{equation*}
    \begin{aligned}
      \forall x,x' \in \R, \qquad [|x-x'|-M]^+ & \leq |x-x'|\ind{|x-x'|\geq M}\\
      & \leq |x-x'|\ind{|x| \geq |x'|\vee M/2} + |x-x'|\ind{|x'| \geq |x|\vee M/2}\\
      & \leq 2|x|\ind{|x|\geq M/2} + 2|x'|\ind{|x'|\geq M/2},
    \end{aligned}
  \end{equation*}
  we obtain
  \begin{equation*}
    J_n^M(t) \leq \int_{u=0}^1 \ex{p(H*\nu^n_t)^{-1}(u)}\ind{\ex{p(H*\nu^n_t)^{-1}(u)} \geq M/2}\dd u + \int_{u=0}^1 \ex{pF_t^{-1}(u)}\ind{\ex{pF_t^{-1}(u)} \geq M/2}\dd u.
  \end{equation*}
  By the Markov inequality,
  \begin{equation*}
    \int_{u=0}^1 \ex{pF_t^{-1}(u)}\ind{\ex{pF_t^{-1}(u)} \geq M/2}\dd u \leq \frac{2}{M}\int_{u=0}^1 \ex{2pF_t^{-1}(u)}\dd u, 
  \end{equation*}
  so that~\eqref{eq:supexpY} applied with $2p$ leads to
  \begin{equation*}
    \lim_{M \to +\infty} \sup_{t \in [0,T]} \int_{u=0}^1 \ex{pF_t^{-1}(u)}\ind{\ex{pF_t^{-1}(u)} \geq M/2}\dd u = 0.
  \end{equation*}
  
  We complete this step by proving that
  \begin{equation}\label{eq:limM}
    \lim_{M \to +\infty} \sup_{n \geq 1} \Exp\left(\sup_{t \in [0,T]} \int_{u=0}^1 \ex{p(H*\nu^n_t)^{-1}(u)}\ind{\ex{p(H*\nu^n_t)^{-1}(u)} \geq M/2}\dd u\right) = 0.
  \end{equation}
  To this aim, we first write
  \begin{equation*}
    \begin{aligned}
      \sup_{t \in [0,T]} \int_{u=0}^1 \ex{p(H*\nu^n_t)^{-1}(u)}\ind{\ex{p(H*\nu^n_t)^{-1}(u)} \geq M/2}\dd u & = \sup_{t \in [0,T]} \frac{1}{n} \sum_{j=1}^n \ex{pY^{j}_n(t)}\ind{\ex{pY^{j}_n(t)} \geq M/2}\\
      & = \sup_{t \in [0,T]} \frac{1}{n} \sum_{i=1}^n \ex{pY^i_n(t)}\ind{\ex{pY^i_n(t)} \geq M/2}\\
      & \leq \frac{1}{n} \sum_{i=1}^n \sup_{t \in [0,T]}  \ex{pY^i_n(t)}\ind{\ex{pY^i_n(t)} \geq M/2},
    \end{aligned}
  \end{equation*}
  so that, owing to the exchangeability of the processes $(Y^1_n(t))_{t \in [0,T]}, \ldots, (Y^n_n(t))_{t \in [0,T]}$,
  \begin{equation*}
    \begin{aligned}
      \Exp\left(\sup_{t \in [0,T]} \int_{u=0}^1 \ex{p(H*\nu^n_t)^{-1}(u)}\ind{\ex{p(H*\nu^n_t)^{-1}(u)} \geq M/2}\dd u\right) & \leq \Exp\left(\sup_{t \in [0,T]}  \ex{pY^1_n(t)}\ind{\ex{pY^1_n(t)} \geq M/2}\right)\\
      & \leq \Exp\left(M_n(T)\ind{M_n(T) \geq M/2}\right),
    \end{aligned}
  \end{equation*}
  where $M_n(T) := \sup_{t \in [0,T]}\ex{pY^1_n(t)}$. Hence, to obtain~\eqref{eq:limM}, it suffices to prove the uniform integrability of the sequence of random variables $(M_n(T))_{n \geq 1}$; thus, it suffices to exhibit $q > p$ such that
  \begin{equation}\label{eq:supq}
    \sup_{n \geq 1} \Exp\left(\sup_{t \in [0,T]} \ex{qY^1_n(t)}\right) < +\infty.
  \end{equation}
  To carry this task out, we fix $q > p$. Proceeding as in the proof~\eqref{eq:supexpY}, we write
  \begin{equation*}
    \begin{aligned}
      \Exp\left(\sup_{t \in [0,T]} \ex{qY^1_n(t)}\right) & \leq \ex{q||\gamma||_{\infty} T} \Exp\left(\ex{q Y^1_n(0)} \sup_{t \in [0,T]}\ex{q \int_{s=0}^t \sigma^1_n(s) \dd B^1(s)}\right)\\
      & = \ex{q||\gamma||_{\infty} T} \Exp\left(\ex{q Y^1_n(0)} \Exp\left(\sup_{t \in [0,T]}\ex{q \int_{s=0}^t \sigma^1_n(s) \dd B^1(s)} \Big| Y^1(0)\right)\right).
    \end{aligned}
  \end{equation*}
  For all $t \in [0,T]$,
  \begin{equation*}
    \ex{q \int_{s=0}^t \sigma^1_n(s) \dd B^1(s)} \leq E(t)^2 \ex{(q^2/4)||\sigma^2||_{\infty}T},
  \end{equation*}
  where $(E(t))_{t \geq 0}$ is the exponential martingale defined by
  \begin{equation*}
    \forall t \geq 0, \qquad E(t) := \ex{(q/2) \int_{s=0}^t \sigma^1_n(s) \dd B^1(s) - (q^2/8)\int_{s=0}^t (\sigma^1_n(s))^2 \dd s}.
  \end{equation*}
  By Doob's inequality,
  \begin{equation*}
    \begin{aligned}
      \Exp\left(\sup_{t \in [0,T]} E(t)^2\Big| Y^1(0)\right) & \leq 4 \Exp\left(E(T)^2\Big| Y^1(0)\right)\\
      & = 4 \Exp\left(\ex{q\int_{s=0}^T \sigma^1_n(s) \dd B^1(s) - (q^2/4)\int_{s=0}^T (\sigma^1_n(s))^2 \dd s}\Big| Y^1(0)\right)\\
      & \leq 4 \Exp\left(\ex{q\int_{s=0}^T \sigma^1_n(s) \dd B^1(s) - (q^2/2)\int_{s=0}^T (\sigma^1_n(s))^2 \dd s}\Big| Y^1(0)\right)\ex{(q^2/4)||\sigma^2||_{\infty} T}\\
      & = 4\ex{(q^2/4)||\sigma^2||_{\infty} T}.
    \end{aligned}
  \end{equation*}
  As a consequence,
  \begin{equation*}
    \Exp\left(\sup_{t \in [0,T]}\ex{q \int_{s=0}^t \sigma^1_n(s) \dd B^1(s)} \Big| Y^1(0)\right) \leq 4\ex{(q^2/2)||\sigma^2||_{\infty} T},
  \end{equation*}
  so that, finally,
  \begin{equation*}
    \begin{aligned}
      \Exp\left(\sup_{t \in [0,T]} \ex{qY^1_n(t)}\right) & \leq 4 \ex{q||\gamma||_{\infty} T + (q^2/2)||\sigma^2||_{\infty}T}\Exp\left(\ex{q Y^1_n(0)}\right)\\
      & = 4 \ex{q||\gamma||_{\infty} T + (q^2/2)||\sigma^2||_{\infty}T} \int_{y \in \R} \ex{qy}m(\dd y).
    \end{aligned}
  \end{equation*}
  By Assumption~\eqref{hyp:H}, the right-hand side above is finite and does not depend on $n$. Therefore,~\eqref{eq:supq} is satisfied and the proof of~\eqref{eq:step2} is completed.
\end{proof}

We complete this appendix with the proof of Lemma~\ref{lem:lln:gamma}.

\begin{proof}[Proof of Lemma~\ref{lem:lln:gamma}]
  We first address the laws of large numbers~\eqref{eq:gammapgammapstar} for the growth rate and excess growth rate. On account of~\eqref{eq:growthrates} and Proposition~\ref{prop:lln:Pi}, it suffices to prove that
  \begin{equation*}
    \sum_{j=1}^n \left(\frac{\ex{p Y^{(j)}_n(t)}}{\ex{p Y^1_n(t)} + \cdots + \ex{p Y^n_n(t)}}\right)^2\sigma^2\left(\frac{j}{n}\right)
  \end{equation*}
  converges to $0$ in $L^q(C([0,T],\R))$. To this aim, we remark that, using the notations of Lemma~\ref{lem:lln}, for all $t \geq 0$,
  \begin{equation*}
    \sum_{j=1}^n \left(\frac{\ex{p Y^{(j)}_n(t)}}{\ex{p Y^1_n(t)} + \cdots + \ex{p Y^n_n(t)}}\right)^2\sigma^2\left(\frac{j}{n}\right) = \frac{1}{n} \frac{\hat{\sigma^2}^{(2p)}_n(t)}{(\hat{\mathbf{1}}^p_n(t))^2},
  \end{equation*}
  where $\mathbf{1}$ refers to the the constant function equal to $1$. By the same arguments as in the proof of Proposition~\ref{prop:lln:Pi} and with the notations of Lemma~\ref{lem:lln}, we obtain that $\hat{\sigma^2}^{(2p)}_n(t)/(\hat{\mathbf{1}}^p_n(t))^2$ converges, in $L^q(C([0,T],\R))$, to $\hat{\sigma^2}^{(2p)}(t)/(\hat{\mathbf{1}}^p(t))^2$, therefore the right-hand side above converges to $0$ and~\eqref{eq:gammapgammapstar} follows.
  
  In addition, we deduce from the argument above and~\eqref{eq:quadvar} that the process $(\langle \log Z_n^p\rangle (t))_{t \in [0,T]}$ converges in probability, in $C([0,T],\R)$, to $0$, and that the process $(\log Z_n^p(t))_{t \in [0,T]}$ converges in probability, in $C([0,T],\R)$, to the process $(\log Z^p(t))_{t \in [0,T]}$ defined by~\eqref{eq:Zpinf}. Using the continuity of the mapping
  \begin{equation*}
    (y(t))_{t \in [0,T]} \mapsto \left(\ex{y(t)}\right)_{t \in [0,T]}
  \end{equation*}
  on $C([0,T],\R)$, we deduce that the process $(Z_n^p(t))_{t \in [0,T]}$ converges in probability, in $C([0,T],\R)$, to the process $(Z^p(t))_{t \in [0,T]}$. Let $q \in [1,+\infty)$. To conclude that the convergence also holds in $L^q(C([0,T],\R))$ we prove that, for $r>q$,
  \begin{equation*}
    \sup_{n \geq 1} \Exp\left(\sup_{t \in [0,T]} |Z_n^p(t)|^r\right) < +\infty.
  \end{equation*}
  This proof of this latter fact is based on Doob's inequality in a similar fashion as at the end of the proof of Lemma~\ref{lem:lln}. Uniformity in $n$ follows from the fact that
  \begin{equation*}
    \langle \log Z_n^p\rangle (t) \leq ||\sigma^2||_{\infty} \int_{s=0}^t \frac{\sum_{i=1}^n \ex{2pY^i_n(s)}}{\left(\sum_{i=1}^n \ex{pY^i_n(s)}\right)^2}\dd s \leq ||\sigma^2||_{\infty} t.
  \end{equation*}
  This completes the proof.
\end{proof}


\section{Long time behaviour of the asymptotic capital measure}\label{app:longtime}

This appendix is dedicated to the proof of Lemma~\ref{lem:bPip} and Proposition~\ref{prop:phtrans}. We first discuss the finiteness of 
\begin{equation*}
  \bZ^p = \int_{u=0}^1 \ex{p \Psi(u)} \dd u.
\end{equation*}

\begin{lem}\label{lem:Zp}
  Let us assume that the uniform ellipticity condtion~\eqref{hyp:UE} and the equilibrium condition~\eqref{hyp:E1} hold, and recall the definition~\eqref{eq:pc} of the critical diversity index $\pc \geq 0$. 
  \begin{itemize}
    \item if $\pc > 0$, then for all $p \in [0,\pc)$, $\bZ^p < +\infty$;
    \item for all $p \in (\pc,+\infty)$, $\bZ^p = +\infty$.
  \end{itemize}
\end{lem}
\begin{proof}
  We shall distinguish between the exponential case $\pc>0$ and the heavy-tailed case $\pc=0$.
  {\em Exponential case, $\pc > 0$}: then $\gamma(1) < g$. Let $\eta > 0$ be small enough for the inequalities $\sigma^2(1)>\eta$ and $g-\gamma(1) > \eta$ to hold. Recall that $\sigma^2(1) > 0$ due to Assumption~\eqref{hyp:UE}. Then, there exists $u^* \in [0,1)$ such that, for all $v \in [u^*,1]$,
  \begin{equation*}
    \begin{aligned}
      & \sigma^2(1)-\eta \leq \sigma^2(v) \leq \sigma^2(1)+\eta,\\
      & g-\gamma(1)-\eta \leq \frac{\Gamma(v)-gv}{1-v} \leq g-\gamma(1)+\eta,
    \end{aligned}
  \end{equation*}
  so that, for all $u \in [u^*,1)$,
  \begin{equation*}
    \begin{aligned}
      \frac{\sigma^2(1)-\eta}{2(g-\gamma(1)+\eta)} \log\left(\frac{1-u^*}{1-u}\right) \leq \Psi(u)-\Psi(u^*) \leq \frac{\sigma^2(1)+\eta}{2(g-\gamma(1)-\eta)} \log\left(\frac{1-u^*}{1-u}\right).
    \end{aligned} 
  \end{equation*}
  As a consequence, for all $p \geq 0$, for all $u \in [u^*,1)$, 
  \begin{equation}\label{eq:cbetaeta}
    C_-(\eta)\int_{v=u^*}^u(1-v)^{-\beta_-(\eta)}\dd v \leq \int_{v=u^*}^u \ex{p\Psi(v)} \dd v \leq C_+(\eta)\int_{v=u^*}^u(1-v)^{-\beta_+(\eta)}\dd v,
  \end{equation}
  where
  \begin{equation*}
    \begin{aligned}
      & \beta_-(\eta) := p\frac{\sigma^2(1)-\eta}{2(g-\gamma(1)+\eta)}, \qquad C_-(\eta) := \ex{p\Psi(u^*)}(1-u^*)^{\beta_-(\eta)},\\
      & \beta_+(\eta) := p\frac{\sigma^2(1)+\eta}{2(g-\gamma(1)-\eta)}, \qquad C_+(\eta) := \ex{p\Psi(u^*)}(1-u^*)^{\beta_+(\eta)}.
    \end{aligned}
  \end{equation*}
  Certainly, $\bZ^p$ is finite if and only if the limit when $u \uto 1$ of the central term in the inequality~\eqref{eq:cbetaeta} is finite.
  \begin{itemize}
    \item If $p \in [0,\pc)$, then for $\eta$ small enough, $\beta_+(\eta) < 1$, so that the right-hand side of~\eqref{eq:cbetaeta} admits a finite limit when $u \uto 1$.
    \item If $p > \pc$, then for $\eta$ small enough, $\beta_-(\eta) > 1$, so that the left-hand side of~\eqref{eq:cbetaeta} grows to $+\infty$ when $u \uto 1$.
  \end{itemize}
  This completes the proof in the case $\pc > 0$.
  
  {\em Heavy-tailed case, $\pc=0$}: then $\gamma(1) = g$. Note that we only have to address the case $p>\pc$. Let $p > 0$ and let $\eta > 0$ small enough for the inequality $\ua p/(2\eta) \geq 1$ to hold, where $\ua := \inf_{u \in [0,1]} \sigma^2(u) > 0$ due to Assumption~\eqref{hyp:UE}. Then, there exists $u^* \in [0,1)$ such that, for all $v \in [u^*,1]$, 
  \begin{equation*}
    \Gamma(v) - gv \leq \eta(1-v),
  \end{equation*}
  so that, for all $u \in [u^*,1)$,
  \begin{equation*}
    \Psi(u) = \Psi(u^*) + \int_{v=u^*}^u \frac{\sigma^2(v)}{2(\Gamma(v)-gv)}\dd v \geq \Psi(u^*) + \frac{\ua}{2\eta}\left(\log(1-u^*)-\log(1-u)\right).
  \end{equation*}
  As a consequence,
  \begin{equation*}
    \ex{p \Psi(u)} \geq \ex{p \Psi(u^*) + \ua \log(1-u^*)/(2\eta)} (1-u)^{-\ua p/(2\eta)},
  \end{equation*}
  and the choice of $\eta$ ensures that the integral of the right-hand side above diverges to $+\infty$ in $1$. This completes the proof in the case $\pc=0$.
\end{proof}

\begin{rk}
  At the criticality, whether $\bZ^{\pc} = +\infty$ or $\bZ^{\pc} < +\infty$ cannot be {\em a priori} determined. Indeed, on the one hand, for the choice of coefficients introduced in Example~\ref{expl:beta}, it is easily checked that $\bZ^{\pc} = +\infty$. On the other hand, assume that $\pc > 0$ and the coefficients $\gamma$ and $\sigma$ are chosen so that the asymptotic expansion of $\Psi$ writes
  \begin{equation*}
    \Psi(u) = \frac{1}{\pc}\left(-\log(1-u) - \beta\log\left(-\log(1-u)\right)\right) + \grandO_{u \uto 1}(1), \qquad \beta > 1.
  \end{equation*}
  Then, it is straightforward to check that $\bZ^{\pc} < +\infty$.
\end{rk}

We now complete the proof of Lemma~\ref{lem:bPip}.

\begin{proof}[Proof of Lemma~\ref{lem:bPip}]
  By Lemma~\ref{lem:Zp}, $\bZ^p < +\infty$ for all $p \in [0, \pc)$, so that the probability distribution $\bPi^p$ is well-defined. We now fix a continuous function $f : [0,1] \to \R$ and prove that the function $p \mapsto \scrochet{f}{\bPi^p}$ is continuous on $[0,\pc)$. Certainly, it suffices to prove that, for all $p \in [0,\pc)$,
  \begin{equation}\label{eq:pf:bPp:1}
    \lim_{p' \to p} \int_{u=0}^1 \ex{p' \Psi(u)} f(u)\dd u = \int_{u=0}^1 \ex{p \Psi(u)} f(u)\dd u.
  \end{equation}
  Let us fix $p \in [0,+\infty)$. Then, for all $u \in (0,1)$,
  \begin{equation*}
    \lim_{p' \to p} \ex{p' \Psi(u)} f(u) = \ex{p \Psi(u)} f(u),
  \end{equation*}
  while, taking $q \in (p,\pc)$, we write
  \begin{equation*}
    \forall p' \in [0,q], \qquad \left|\ex{p' \Psi(u)}f(u)\right| \leq \ex{q [\Psi(u)]^+} ||f||_{\infty},
  \end{equation*}
  where we recall that $[\psi]^+ := \psi \vee 0$. It easily follows from Lemma~\ref{lem:Zp} that the right-hand side above is integrable on $[0,1]$, so that~\eqref{eq:pf:bPp:1} stems from the dominated convergence theorem. Note that the same arguments allow to prove that, if $\bZ^{\pc} < +\infty$, then the function $p \mapsto \scrochet{f}{\bPi^p}$ is continuous on $[0,\pc]$.
  
  To complete the proof, it remains to show that, if $\bZ^{\pc} = +\infty$, then $\lim_{p \uto \pc} \scrochet{f}{\bPi^p} = f(1)$. In this purpose, we assume that $\bZ^{\pc} = +\infty$. Then, Fatou's lemma immediately yields
  \begin{equation*}
    \lim_{p \uto \pc} \int_{u=0}^1 \ex{p \Psi(u)}\dd u = +\infty.
  \end{equation*}
  Let $\eta > 0$, then by the continuity of $f$, there exists $u^* \in [0,1)$ such that, for all $u \in [u^*,1]$, $f(1)-\eta \leq f(u) \leq f(1)+\eta$. Let us define, for all $p \in [0,\pc)$,
  \begin{equation*}
    I_f(p) := \frac{\displaystyle \int_{u=0}^{u^*} \ex{p \Psi(u)}f(u)\dd u}{\displaystyle \int_{u=0}^1 \ex{p \Psi(u)}\dd u}, \qquad I(p) := \frac{\displaystyle \int_{u=0}^{u^*} \ex{p \Psi(u)}\dd u}{\displaystyle \int_{u=0}^1 \ex{p \Psi(u)}\dd u},
  \end{equation*}
  and write
  \begin{equation*}
    \crochet{f}{\bPi^p} = I_f(p) + \frac{\displaystyle \int_{u=u^*}^1 \ex{p \Psi(u)}f(u)\dd u}{\displaystyle \int_{u=0}^1 \ex{p \Psi(u)}\dd u},
  \end{equation*}
  so that
  \begin{equation*}
    I_f(p) + (f(1)-\eta)(1-I(p)) \leq \crochet{f}{\bPi^p} \leq I_f(p) + (f(1)+\eta)(1-I(p)).
  \end{equation*}
  
  Observing that
  \begin{equation*}
    \limsup_{p \uto \pc} \left|\int_{u=0}^{u^*} \ex{p \Psi(u)}f(u)\dd u\right| \leq ||f||_{\infty} \int_{u=0}^{u^*} \ex{\pc \Psi(u)}\dd u < +\infty,
  \end{equation*}
  we deduce that $I_f(p)$ and $I(p)$ vanish when $p \uto \pc$. The conclusion is straightforward.
\end{proof}

We now prove Proposition~\ref{prop:phtrans}. We shall use the following result regarding the convergence in Wasserstein distance.

\begin{lem}\label{lem:wass}
  Let $(G_t)_{t \geq 0}$ be a family of cumulative distribution functions on $\R$ and $G_{\infty}$ be a cumulative distribution function on $\R$, such that:
  \begin{itemize}
    \item there exists $q \geq 1$ such that $\lim_{t \to +\infty} W_q(G_t, G_{\infty}) = 0$,
    \item the probability distribution with cumulative distribution function $G_{\infty}$ admits a positive density with respect to the Lebesgue measure on $\R$, so that both $G_{\infty}$ and $G^{-1}_{\infty}$ are continuous, respectively on $\R$ and $(0,1)$.
  \end{itemize}
  Then, for all $y \in \R$, $\lim_{t \to +\infty} G_t(y) = G_{\infty}(y)$, and for all $u \in (0,1)$, $\lim_{t \to +\infty} G_t^{-1}(u) = G_{\infty}^{-1}(u)$.
\end{lem}
\begin{proof}
  Since the Wasserstein distance metrizes the weak convergence, $G_t$ converges weakly to $G_{\infty}$. This classically implies that $G_t(y)$ converges to $G_{\infty}(y)$ for all continuity point $y \in \R$ of $G_{\infty}$~\cite[Theorem~2.2, p.~86]{durrett}, and $G_t^{-1}(u)$ converges to $G_{\infty}^{-1}(u)$ for all continuity point $u \in (0,1)$ of $G_{\infty}^{-1}$~\cite[Theorem~2.1, p.~85]{durrett}. Since $G_{\infty}$ is continuous on $\R$ and $G_{\infty}^{-1}$ is continuous on $(0,1)$, then the proof is completed. 
\end{proof}

\begin{proof}[Proof of Proposition~\ref{prop:phtrans}]
  Let us assume that the conditions of Theorem~\ref{theo:equi} and Proposition~\ref{prop:lln:Pi} are satisfied. Recall that the critical diversity index $\pc \geq 0$ is defined in~\eqref{eq:pc}.

  {\em Subcritical case.} Let us assume that $\pc > 0$ and let $p \in [0,\pc)$. Following~\S\ref{sss:heuristic}, it suffices to prove that, for all continuous function $f : [0,1] \to \R$,
  \begin{equation}\label{eq:limtexpFt}
    \lim_{t \to +\infty} \int_{u=0}^1 \ex{p \tF^{-1}_t(u)} f(u) \dd u = \int_{u=0}^1 \ex{p \tF^{-1}_{\infty}(u)} f(u) \dd u,
  \end{equation}
  where the cumulative distribution function $\tF_{\infty}$ is defined by Theorem~\ref{theo:equi}. Combining the latter with Lemma~\ref{lem:wass}, it is already known that, for all $u \in (0,1)$, $\lim_{t \to +\infty} \tF_t^{-1}(u) = \tF_{\infty}^{-1}(u)$. As a consequence, and since $f$ is bounded,~\eqref{eq:limtexpFt} follows if we exhibit $q > p$ such that
  \begin{equation}\label{eq:supqFt}
    \sup_{t \geq 0} \int_{u=0}^1 \ex{q \tF_t^{-1}(u)} \dd u < +\infty.
  \end{equation}
  In this purpose, let us fix $q \geq 0$ such that $p < q < \pc$ and remark that
  \begin{equation*}
    \int_{u=0}^1 \ex{q \tF_t^{-1}(u)} \dd u = \Exp\left(\ex{q \tY(t)}\right).
  \end{equation*}
  By Itô's formula and~\eqref{eq:supexpY},
  \begin{equation*}
    \frac{\dd}{\dd t} \Exp\left(\ex{q \tY(t)}\right) = \Exp\left(\ex{q \tY(t)}\tilde{b}_q(\tF_t(\tY(t)))\right),
  \end{equation*}
  where $\tilde{b}_q(u) := q \tilde{\gamma}(u) + q^2\sigma^2(u)/2$. It follows from a straightforward analysis of the function $q \mapsto \tilde{b}_q(1)$ that, since $q < \pc$, then $\tilde{b}_q(1) < 0$. Therefore, by the continuity of $\tilde{b}_q$, there exist $\eta > 0$ and $u^* \in [0,1)$ such that, for all $u \in [u^*,1]$, $\tilde{b}_q(u) \leq -\eta$. As a consequence, for all $t \geq 0$,
  \begin{equation*}
    \begin{aligned}
      & \Exp\left(\ex{q \tY(t)}\tilde{b}_q(\tF_t(\tY(t)))\right)\\
      & \qquad = \Exp\left(\ex{q \tY(t)}\tilde{b}_q(\tF_t(\tY(t)))\ind{\tF_t(\tY(t)) < u^*}\right) + \Exp\left(\ex{q \tY(t)}\tilde{b}_q(\tF_t(\tY(t)))\ind{\tF_t(\tY(t)) \geq u^*}\right)\\
      & \qquad \leq \Exp\left(\ex{q \tY(t)}\tilde{b}_q(\tF_t(\tY(t)))\ind{\tF_t(\tY(t)) < u^*}\right) - \eta \Exp\left(\ex{q \tY(t)}\ind{\tF_t(\tY(t)) \geq u^*}\right)\\
      & \qquad \leq (||\tilde{b}_q||_{\infty} + \eta) \Exp\left(\ex{q \tY(t)}\ind{\tF_t(\tY(t)) < u^*}\right) - \eta \Exp\left(\ex{q \tY(t)}\right).
    \end{aligned}
  \end{equation*}
  For all $t \geq 0$, the definition of $\tF_t^{-1}$ and the right continuity of $\tF_t$ yield, for all $u \in (0,1)$, $\tF_t(\tF_t^{-1}(u)) \geq u$. As a consequence,
  \begin{equation*}
    \Exp\left(\ex{q \tY(t)}\ind{\tF_t(\tY(t)) < u^*}\right) = \int_{u=0}^1 \ex{q\tF_t^{-1}(u)}\ind{\tF_t(\tF_t^{-1}(u)) < u^*} \dd u \leq \int_{u=0}^{u^*} \ex{q\tF_t^{-1}(u)} \dd u \leq u^* \ex{q\tF_t^{-1}(u^*)},
  \end{equation*}
  and the right-hand side converges to $u^* \ex{q\tF_{\infty}^{-1}(u^*)} < +\infty$ when $t \to +\infty$. As a consequence, there exists $C < +\infty$ such that
  \begin{equation*}
    \frac{\dd}{\dd t} \Exp\left(\ex{q \tY(t)}\right) \leq C - \eta \Exp\left(\ex{q \tY(t)}\right),
  \end{equation*}
  and~\eqref{eq:supqFt} follows from Gronwall's lemma.
  
  {\em Supercritical case.} For $p > \pc \geq 0$, Theorem~\ref{theo:equi}, Fatou's lemma and Lemma~\ref{lem:Zp} yield
  \begin{equation}\label{eq:pf:lemltsuper0}
    \lim_{t \to +\infty} \int_{u=0}^1 \ex{p \tF_t^{-1}(u)}\dd u = +\infty.
  \end{equation}
  Let $\eta > 0$. By the continuity of $f$, there exists $u^* \in [0,1)$ such that, for all $u \in [u^*,1]$, $f(1) - \eta \leq f(u) \leq f(1) + \eta$. Besides, there exists $M > 0$ such that $\tF_{\infty}(M) > u^*$. Then, for all $t \geq 0$,
  \begin{equation*}
    \begin{aligned}
      \int_{u=0}^1 \ex{p\tF_t^{-1}(u)}f(u) \dd u & = \Exp\left(\ex{p Y(t)}f(\tF_t(\tY(t)))\right)\\
      & = \Exp\left(\ex{p \tY(t)}f(\tF_t(\tY(t)))\ind{\tY(t) < M}\right) + \Exp\left(\ex{p \tY(t)}f(\tF_t(\tY(t)))\ind{\tY(t) \geq M}\right).
    \end{aligned}
  \end{equation*}
  
  On the one hand,
  \begin{equation*}
    \Exp\left(\ex{p \tY(t)}f(\tF_t(\tY(t)))\ind{\tY(t) < M}\right) \leq ||f||_{\infty} \ex{pM},
  \end{equation*}
  so that
  \begin{equation*}
    \lim_{t \to +\infty} \frac{\Exp\left(\ex{p \tY(t)}f(\tF_t(\tY(t)))\ind{\tY(t) < M}\right)}{\Exp\left(\ex{p \tY(t)}\right)} = 0.
  \end{equation*}
  
  On the other hand, since Lemma~\ref{lem:wass} implies that $\lim_{t \to +\infty} \tF_t(M) = \tF_{\infty}(M)$, then for $t$ large enough one has, for all $y \geq M$, $\tF_t(y) \geq \tF_t(M) \geq u^*$. Therefore, for $t$ large enough,
  \begin{equation*}
    f(1)-\eta \leq \frac{\Exp\left(\ex{p \tY(t)}f(\tF_t(\tY(t)))\ind{\tY(t) \geq M}\right)}{\Exp\left(\ex{p \tY(t)}\ind{\tY(t) \geq M}\right)} \leq f(1)+\eta,
  \end{equation*}
  while $\Exp\left(\ex{p \tY(t)}\ind{\tY(t) \geq M}\right) / \Exp\left(\ex{p \tY(t)}\right)$ converges to $1$. As a conclusion, 
  \begin{equation*}
    f(1) - \eta \leq \liminf_{t \to +\infty} \frac{\displaystyle \int_{u=0}^1 \ex{p\tF_t^{-1}(u)} f(u)\dd u}{\displaystyle \int_{u=0}^1 \ex{p\tF_t^{-1}(u)}\dd u} \leq \limsup_{t \to +\infty} \frac{\displaystyle \int_{u=0}^1 \ex{p\tF_t^{-1}(u)} f(u)\dd u}{\displaystyle \int_{u=0}^1 \ex{p\tF_t^{-1}(u)}\dd u} \leq f(1) + \eta,
  \end{equation*}
  and the proof of the supercritical case is completed.  
  
  {\em Criticality, case $\bZ^{\pc} = +\infty$.} Note that the proof in the supercritical case above only requires that $p$ be such that~\eqref{eq:pf:lemltsuper0} holds. As soon as $\bZ^{\pc} = +\infty$, Fatou's lemma implies that~\eqref{eq:pf:lemltsuper0} holds with $p=\pc$, so that we similarly obtain that $\lim_{t \to +\infty} \scrochet{f}{\Pi^{\pc}(t)} = f(1)$.

  {\em Criticality, case $\bZ^{\pc} < +\infty$.} We finally assume that $\bZ^{\pc} < +\infty$ and prove~\eqref{eq:encadre}. In this purpose, we let $\ell \in [-||f||_{\infty}, ||f||_{\infty}]$ be the limit of a converging sequence $(\scrochet{f}{\Pi^{\pc}(t_k)})_{k \geq 1}$, where $t_k$ grows to infinity with $k$. We shall prove that
  \begin{equation}\label{eq:encadre:ell}
    f(1) \wedge \crochet{f}{\bPi^{\pc}} \leq \ell \leq f(1) \vee \crochet{f}{\bPi^{\pc}}.
  \end{equation}
  
  First, we deduce from Fatou's lemma that there exists a subsequence of $(t_k)_{k \geq 1}$, that we still index by $k$ for convenience, such that
  \begin{equation*}
    \lim_{k \to +\infty} \int_{u=0}^1 \ex{\pc \tF_{t_k}^{-1}(u)} \dd u = I \in [J, +\infty],
  \end{equation*}
  where
  \begin{equation*}
    J := \int_{u=0}^1\ex{\pc \tF^{-1}_{\infty}(u)}\dd u < +\infty.
  \end{equation*}
  
  Let us now fix $\eta > 0$. By the continuity of $f$, there exists $u^* \in [0,1)$ such that, for all $u \in [u^*,1]$, $f(1)-\eta \leq f(u) \leq f(1)+\eta$. Now let $M \geq 0$ be large enough for the inequality $\tF_{\infty}(M) > u^*$ to hold. Then, for all $k \geq 1$,
  \begin{equation*}
    \crochet{f}{\Pi^{\pc}(t_k)} = \frac{\displaystyle \int_{u=0}^1 \ex{\pc \tF^{-1}_{t_k}(u)}f(u)\ind{\tF^{-1}_{t_k}(u) \leq M}\dd u}{\displaystyle \int_{u=0}^1 \ex{\pc \tF^{-1}_{t_k}(u)}\dd u} + \frac{\displaystyle \int_{u=0}^1 \ex{\pc \tF^{-1}_{t_k}(u)}f(u)\ind{\tF^{-1}_{t_k}(u) > M}\dd u}{\displaystyle \int_{u=0}^1 \ex{\pc \tF^{-1}_{t_k}(u)}\dd u}.
  \end{equation*}
  
  On the one hand, since the equilibrium distribution does not weight points,
  \begin{equation*}
    \lim_{k \to +\infty} \int_{u=0}^1 \ex{\pc \tF^{-1}_{t_k}(u)}f(u)\ind{\tF^{-1}_{t_k}(u) \leq M}\dd u = \int_{u=0}^1 \ex{\pc \tF^{-1}_{\infty}(u)}f(u)\ind{\tF^{-1}_{\infty}(u) \leq M}\dd u =: J^M_f,
  \end{equation*}
  and the limit is finite. As a consequence,
  \begin{equation*}
    \lim_{k \to +\infty} \frac{\displaystyle \int_{u=0}^1 \ex{\pc \tF^{-1}_{t_k}(u)}f(u)\ind{\tF^{-1}_{t_k}(u) \leq M}\dd u}{\displaystyle \int_{u=0}^1 \ex{\pc \tF^{-1}_{t_k}(u)}\dd u} = \frac{J^M_f}{I},
  \end{equation*}
  where it is understood that the limit is null whenever $I=+\infty$.
  
  On the other hand, by Lemma~\ref{lem:wass}, for $k$ large enough, $F_{t_k}(M) \geq u^*$ so that 
  \begin{equation*}
    \begin{aligned}
      & (f(1)-\eta)\int_{u=0}^1 \ex{\pc \tF^{-1}_{t_k}(u)}\ind{\tF^{-1}_{t_k}(u) > M}\dd u\\
      & \qquad \leq \int_{u=0}^1 \ex{\pc \tF^{-1}_{t_k}(u)}f(u)\ind{\tF^{-1}_{t_k}(u) > M}\dd u \leq (f(1)+\eta)\int_{u=0}^1 \ex{\pc \tF^{-1}_{t_k}(u)}\ind{\tF^{-1}_{t_k}(u) > M}\dd u,
    \end{aligned}
  \end{equation*}
  therefore
  \begin{equation*}
    \begin{aligned}
      & (f(1)-\eta)\left(1-\frac{\displaystyle\int_{u=0}^1 \ex{\pc \tF^{-1}_{t_k}(u)}\ind{\tF^{-1}_{t_k}(u) \leq M}\dd u}{\displaystyle \int_{u=0}^1 \ex{\pc \tF^{-1}_{t_k}(u)}\dd u}\right)\\
      & \qquad \leq \frac{\displaystyle \int_{u=0}^1 \ex{\pc \tF^{-1}_{t_k}(u)}f(u)\ind{\tF^{-1}_{t_k}(u) > M}\dd u}{\displaystyle \int_{u=0}^1 \ex{\pc \tF^{-1}_{t_k}(u)}\dd u} \leq (f(1)+\eta)\left(1-\frac{\displaystyle\int_{u=0}^1 \ex{\pc \tF^{-1}_{t_k}(u)}\ind{\tF^{-1}_{t_k}(u) \leq M}\dd u}{\displaystyle \int_{u=0}^1 \ex{\pc \tF^{-1}_{t_k}(u)}\dd u}\right),
    \end{aligned}
  \end{equation*}
  
  As a consequence,
  \begin{equation*}
    \frac{J^M_f}{I} + (f(1)-\eta)\left(1-\frac{J^M}{I}\right) \leq \ell \leq \frac{J^M_f}{I} + (f(1)+\eta)\left(1-\frac{J^M}{I}\right) 
  \end{equation*}
  where
  \begin{equation*}
    J^M := \lim_{k \to +\infty} \int_{u=0}^1 \ex{\pc \tF^{-1}_{t_k}(u)}\ind{\tF^{-1}_{t_k}(u) \leq M}\dd u = \int_{u=0}^1 \ex{\pc \tF^{-1}_{\infty}(u)}\ind{\tF^{-1}_{\infty}(u) \leq M}\dd u.
  \end{equation*}
  By the dominated convergence theorem,
  \begin{equation*}
    \lim_{M \to +\infty} J^M_f = J_f := \int_{u=0}^1 \ex{\pc \tF^{-1}_{\infty}(u)}f(u)\dd u, \qquad \lim_{M \to +\infty} J^M = J,
  \end{equation*}
  so that
  \begin{equation*}
    \frac{J_f}{I} + (f(1)-\eta)\left(1-\frac{J}{I}\right) \leq \ell \leq \frac{J_f}{I} + (f(1)+\eta)\left(1-\frac{J}{I}\right),
  \end{equation*}
  and letting $\eta$ vanish yields
  \begin{equation*}
    \ell = \frac{J_f}{I} + f(1)\left(1-\frac{J}{I}\right).
  \end{equation*}

  We conclude by remarking that $J/I \in [0,1]$, while 
  \begin{equation*}
    \frac{J_f}{I} = \frac{J_f}{J}\frac{J}{I} = \crochet{f}{\bPi^{\pc}}\frac{J}{I},
  \end{equation*}
  so that $\ell$ writes as a convex combination of $\scrochet{f}{\bPi^{\pc}}$ and $f(1)$ and therefore satisfies~\eqref{eq:encadre:ell}.
\end{proof}


\end{document}